\documentclass{aic}
\usepackage{graphicx} 
\usepackage{macros_journal} 
\usetikzlibrary{shapes.geometric, arrows}
\usepackage{color}
\hypersetup{citecolor=magenta}

\def\RS{\mathsf{RS}}
\aicAUTHORdetails{%
  title = {Random Reed--Solomon Codes Achieve List-Decoding Capacity With Linear-Sized Alphabets}, 
  author = {Omar Alrabiah, Zeyu Guo, Venkatesan Guruswami, Ray Li, and Zihan Zhang},
  plaintextauthor = {Omar Alrabiah, Zeyu Guo, Venkatesan Guruswami, Ray Li, Zihan Zhang},
  plaintexttitle = {Random Reed--Solomon Codes Achieve List-Decoding Capacity With Linear-Sized Alphabets}, 
  runningtitle = {Random RS Codes Achieve List-Decoding Capacity With Linear-Sized Alphabets}, 
  keywords = {Reed--Solomon codes, list decoding, error-correcting codes},
}   

\aicEDITORdetails{%
   year={2025},
   number={8},
   received={13 June 2024},   
   published={3 October 2025},  
   doi={10.19086/aic.2025.8},      
}   

\begin{document}

\begin{frontmatter}[classification=text]

\title{Random Reed--Solomon Codes Achieve List-Decoding Capacity With Linear-Sized Alphabets
\titlefootnote{Preliminary versions of these results appeared at the conferences FOCS '23 and STOC '24~\cite{GZ23,AGL24b}.}} 

\author[omara]{Omar Alrabiah\thanks{Supported in part by a Saudi Arabian Cultural Mission (SACM) Scholarship, NSF CCF-2210823 and V.\ Guruswami's Simons Investigator Award.}}
\author[zeyug]{Zeyu Guo\thanks{Supported by NSF grant CCF-2440926.}}
\author[venkatg]{Venkatesan Guruswami\thanks{Supported by a Simons Investigator Award and NSF grants CCF-2210823 and CCF-2228287.}}
\author[rayl]{Ray Li\thanks{Supported by a NSF Mathematical Sciences Postdoctoral Research Fellowships Program under Grant DMS-2203067, and a UC Berkeley Initiative for Computational Transformation award.}}
\author[zihanz]{Zihan Zhang\thanks{Supported in part by NSF grant CCF-2440926.}}

\begin{abstract}
Reed--Solomon codes are a classic family of error-correcting codes consisting of evaluations of low-degree polynomials over a finite field on some sequence of distinct field elements. They are widely known for their optimal unique-decoding capabilities, but their list-decoding capabilities are not fully understood. Given the prevalence of Reed-Solomon codes, a fundamental question in coding theory is determining if Reed--Solomon codes can optimally achieve list-decoding capacity.
  
  A recent breakthrough by Brakensiek, Gopi, and Makam established that Reed--Solomon codes are combinatorially list-decodable all the way to capacity. However, their results hold for randomly-punctured Reed--Solomon codes over an exponentially large field size $2^{O(n)}$, where $n$ is the block length of the code. A natural question is whether Reed--Solomon codes can still achieve capacity over smaller fields. 
  We show that Reed--Solomon codes are list-decodable to capacity with linear field size $O(n)$, which is evidently optimal up to a constant factor. 

  Our techniques also show that random linear codes are list-decodable up to capacity with optimal list-size $O(1/\varepsilon)$ and near-optimal alphabet size $2^{O(1/\varepsilon^2)}$, where $\varepsilon$ is the gap to capacity. As far as we are aware, list-decoding up to capacity with optimal list-size $O(1/\varepsilon)$ was not known to be achievable with any linear code over a constant alphabet size (even non-constructively), and it was also not known to be achievable for random linear codes over any alphabet size.

  With our proof, which maintains a hypergraph perspective of the list-decoding problem, we include an alternate presentation of ideas from Brakensiek, Gopi, and Makam that more directly connects the list-decoding problem to the GM-MDS theorem via a hypergraph orientation theorem.
\end{abstract}
\end{frontmatter}

\section{Introduction}

An (\emph{error-correcting}) \emph{code} is simply a set of strings (\emph{codewords}).
In this paper, all codes are \emph{linear}, meaning our code $C \subseteq \mathbb{F}_q^n$ is a space of vectors over a finite field $\mathbb{F}_q$, for some prime power $q$. 
A \emph{Reed--Solomon code} \cite{RS60} is a linear code obtained by evaluating low-degree polynomials over $\mathbb{F}_q$. More formally,
\begin{align}
  \RS_{n,k}(\alpha_1,\dots,\alpha_n) \defeq \{(f(\alpha_1),\dots,f(\alpha_n))\in\mathbb{F}_q^n : f\in \mathbb{F}_q[X], \deg(f) < k\}.
\end{align}
The \emph{rate} $R$ of a code $C$ is $R\defeq \log_q|C|/n$, which, for a Reed--Solomon code, is $k/n$.
Famously, Reed--Solomon codes are optimal for the \emph{unique decoding problem} \cite{RS60}: for any rate $R$ Reed--Solomon code, for every received word $y\in\mathbb{F}_q^n$, there is at most one codeword within Hamming distance $pn$ of $y$ for \emph{error parameter} $p=(1-R)/2$,\footnote{The Hamming distance between two codewords is the number of coordinates on which they differ.} and further this error parameter $p=\frac{1-R}{2}$ is optimal by the \emph{Singleton bound} \cite{Singleton64}.

In this paper, we study Reed--Solomon codes in the context of \emph{list-decoding}, a generalization of unique-decoding that was introduced by Elias and Wozencraft \cite{Elias57,Wozencraft58}. 
Formally, a code $C \subseteq \mathbb{F}_q^n$ is \emph{$(p,L)$-list-decodable} if, for every received word $y\in\mathbb{F}_q^n$, there are at most $L$ codewords of $C$ within Hamming distance $pn$ of $y$. 

It is well known that the largest fraction of errors that can be list-decoded with small lists approaches the quantity $1-R$ \cite[Theorem 7.4.1]{GRS22}. Specifically, for $p=1-R-\varepsilon$, there are (infinite families of) rate $R$ codes that are $(p,L)$ list-decodable for a list-size $L$ as small as $O(1/\varepsilon)$. On the other hand, for $p=1-R+\varepsilon$, if a rate $R$ code is $(p,L)$ list decodable, the list size $L$ must be exponential in the code length $n$. The quantity $1-R$ is therefore referred to as the \emph{list-decoding capacity}, to characterize the limiting fraction of errors that can be list-decoded as a function of the code rate.
Informally, a code that is list-decodable up to radius $p=1-R-\varepsilon$ with list size $O_\varepsilon(1)$, or even list size $n^{O_\varepsilon(1)}$ where $n$ is the code length, is said to \emph{achieve} (\emph{list-decoding}) \emph{capacity}.

The list-decodability of Reed--Solomon codes is important for several reasons.
Reed--Solomon codes are the most fundamental algebraic error-correcting codes. 
In fact, all of the prior explicit constructions of codes achieving list-decoding capacity are based on algebraic constructions that generalize Reed--Solomon codes, for example, Folded Reed--Solomon codes~\cite{GuruswamiR08,KRSW18}, Multiplicity codes~\cite{GW13,Kop15,KRSW18}, and algebraic-geometric codes~\cite{DL12,GX12,GX13,HRW19}.
Thus, it is natural to wonder whether and when Reed--Solomon codes themselves achieve list-decoding capacity.
Additionally, all Reed--Solomon codes are optimally \emph{unique-decodable}, so (equivalently) they are optimally list-decodable with the list size $L=1$, making them a natural candidate for codes achieving list-decoding capacity. 
Further, capacity-achieving Reed--Solomon codes would potentially offer advantages over existing explicit capacity-achieving codes, such as simplicity and potentially smaller alphabet sizes (which we achieve in this work).
Lastly, list-decoding of Reed–Solomon codes has found several applications in complexity theory and pseudorandomness \cite{CPS99, STV01,LP20}.

For all these reasons, the list-decodability of Reed--Solomon codes is well-studied. As rate $R$ Reed--Solomon codes are uniquely decodable up to the optimal radius $\frac{1-R}{2}$ given by the Singleton Bound, the Johnson-bound \cite{Johnson62} automatically implies that Reed--Solomon codes are $(p,L)$-list-decodable for error parameter $p=1-\sqrt{R}-\varepsilon$ and list size $L=O(1/\varepsilon)$.
Guruswami and Sudan \cite{GuruswamiS99} showed how to \emph{efficiently} list-decode Reed--Solomon codes up to the Johnson radius $1-\sqrt{R}$.
For a long time, this remained the best list-decodability result (even non-constructively) for Reed--Solomon codes.

Since then, several results suggested Reed--Solomon codes could \emph{not} be list-decoded up to capacity, and in fact, not much beyond the Johnson radius $1-\sqrt{R}$.
Guruswami and Rudra \cite{GR06} showed that, for a generalization of list-decoding called \emph{list-recovery}, Reed--Solomon codes are not list-recoverable beyond the (list-recovery) Johnson bound in some parameter settings.
Cheng and Wan \cite{CW07} showed that efficient list-decoding of Reed--Solomon codes beyond the Johnson radius in certain parameter settings implies fast algorithms for the discrete logarithm problem.
Ben-Sasson, Kopparty, and Radhakrishnan~\cite{BKR10} showed that full-length Reed--Solomon codes ($q=n$) are not list-decodable much beyond the Johnson bound in some parameter settings.

Nevertheless, a subsequent exciting line of work \cite{RudraW14, ST20, GLSTW21, FKS22, GST22, BGM23} has shown the existence of Reed--Solomon codes that could in fact be list-decoded beyond the Johnson radius.
These works all consider \emph{combinatorial} list-decodability of \emph{randomly punctured} Reed--Solomon codes.
By combinatorial list-decodability, we mean that the code is proved to be list-decodable without providing an algorithm to efficiently decode the list of nearby codewords.
By randomly punctured Reed--Solomon code, we mean a code $\RS_{n,k}(\alpha_1,\dots,\alpha_n)$ where $(\alpha_1,\dots,\alpha_n)$ are chosen uniformly over all $n$-tuples of pairwise distinct elements of $\mathbb{F}_q$. 
Several of these works~\cite{RudraW14, FKS22, GST22} proved more general list-decoding results about randomly puncturing any code with good unique-decoding properties, not just Reed--Solomon codes.

In this line of work, a recent breakthrough of Brakensiek, Gopi, and Makam~\cite{BGM23} showed, using notions of ``higher-order MDS codes'' \cite{BGM22,Roth22}, that Reed--Solomon codes can actually be list-decoded up to capacity.
In fact, they show, more strongly, that Reed--Solomon codes can be list-decoded with list size $L$ with radius $p=\frac{L}{L+1}(1-R)$, exactly meeting the \emph{generalized Singleton bound} \cite{ST20}, resolving a conjecture of Shangguan and Tamo \cite{ST20}. 
However, their results require randomly puncturing Reed--Solomon codes over an exponentially large field size $2^{O(n)}$, where $n$ is the block length of the code.

A natural question is how small we can take the field size in a capacity-achieving Reed--Solomon code.
It was shown \cite{BDG23,AGL24} that the exponential-in-$n$ field size in \cite{BGM23} is indeed necessary to \emph{exactly} achieve the generalized Singleton bound\footnote{In \cite{BDG23}, this was shown for $L=2$ under the additional assumptions that the code is linear and MDS, and the general statement was later proved in \cite{AGL24}.} but smaller field sizes remained possible if one allowed a small $\varepsilon$ slack in the parameters.

  \subsection{Our Results}
  \label{subsec:results}

  \paragraph{List-decoding Reed--Solomon codes.}

  We show that Reed--Solomon codes are list-decodable up to capacity and the generalized Singleton bound with linear alphabet size $O(n)$, which is evidently optimal up to a constant factor.
  Our main result is the following.

\begin{theorem}
  Let $\varepsilon\in(0,1)$, $L\ge 2$ and $q$ be a prime power such that $q\ge n+k\cdot 2^{10L/\varepsilon}$. Then with probability at least $1-2^{-Ln}$, a randomly punctured Reed--Solomon code of block length $n$ and rate $k/n$ over $\mathbb{F}_q$ is $(\frac{L}{L+1}(1-R-\varepsilon),L)$ average-radius list-decodable.
  \label{thm:main}
\end{theorem}

As in previous works like \cite{BGM23}, Theorem~\ref{thm:main} gives \emph{average-radius list-decodability}, a stronger guarantee than list-decodability: for any distinct $L+1$ codewords $c\ind{1},\dots,c\ind{L+1}$ and any vector $y\in\mathbb{F}_q^n$, the average Hamming distance from $c\ind{1},\dots,c\ind{L+1}$ to $y$ is at least $\frac{L}{L+1}(1-R-\varepsilon)$.
Taking $L=O(1/\epsilon)$ in Theorem~\ref{thm:main}, it follows that Reed--Solomon codes achieve list-decoding capacity even over linear-sized alphabets.
\begin{corollary}
  Let $\varepsilon\in(0,1)$ and $q$ be a prime power such that $q\ge n+k\cdot 2^{O(1/\varepsilon^2)}$. Then with probability at least $1-2^{-\Omega(n/\varepsilon)}$, a randomly punctured Reed--Solomon code of block length $n$ and rate $k/n$ over $\mathbb{F}_q$ is $\left(1-R-\varepsilon,O\left(\frac{1}{\varepsilon}\right)\right)$ average-radius list-decodable.
  \label{cor:main}
\end{corollary}

In our proof of Theorem~\ref{thm:main}, we maintain a hypergraph perspective of the list-decoding problem, which was introduced in \cite{GLSTW21}. Section~\ref{ssec:hypergraph} elaborates on the advantages of this perspective, which include (i) more compact notations, definitions, and lemma statements, (ii) some more streamlined proofs, and (iii) an alternate presentation of ideas from Brakensiek, Gopi, and Makam \cite{BGM23} that more directly connects the list-decoding problem to the so-called GM-MDS theorem \cite{DauSY14,Lovett18,YildizH19} via a hypergraph orientation theorem (see Appendix~\ref{app:bgm}).  

\paragraph{List-decoding random linear codes.}
A random linear code of rate $R$ and length $n$ over $\mathbb{F}_q$ is a random subspace of $\mathbb{F}_q^n$ of dimension $Rn$.
List-decoding random linear codes is well-studied \cite{ZP81, Elias91, GHSZ02,GHK11,Wootters13,RudraW14,RudraW18, LiW20,MRRSW20, GLMSSW21, GM22,PP24} and is an important question for several reasons.
First, finding explicit codes approaching list-decoding capacity is a major challenge, and random linear codes provide a stepping stone towards explicit codes: it is easily seen that uniformly random codes achieve list-decoding capacity, and showing list-decodability of random linear codes can be viewed as a derandomization of the uniformly random construction (see \cite{Elias91} for a discussion of the challenge of showing list-decodability of linear codes, as was first done in \cite{ZP81}). 
Mathematically, the list-decodability of random linear codes concerns a fundamental geometric question: to what extent do random subspaces over $\mathbb{F}_q$ behave like uniformly random sets?
In coding theory, list-decodable random linear codes are useful building blocks in other coding theory constructions \cite{GI01, HW18}.
Lastly, the algorithmic question of decoding random linear codes is closely related to the Learning With Errors (LWE) problem in cryptography \cite{Regev09} and Learning Parity with Noise (LPN) problem in learning theory \cite{BKW03,FGKP06}.

The list-decodability of random linear codes is more difficult to analyze than uniformly random codes, because codewords do not enjoy the same independence as in random codes. Thus the naive argument that shows that random linear codes achieve list-decoding capacity \cite{ZP81} gives an exponentially worse list size of $q^{1/\varepsilon}$ than for random codes ($\varepsilon$ is the gap to the ``$q$-ary capacity'', $R=1-H_q(p)$,  where $H_q(x)\defeq x\log_q(q-1) - x\log_q(x)-(1-x)\log_q(1-x)$ is the $q$-ary entropy function). 
Several works have sought to circumvent this difficulty \cite{Elias91, GHSZ02, GHK11, Wootters13, RudraW14, RudraW18, LiW20, GLMSSW21} improving the list-size bound to $O_q(1/\varepsilon)$, matching the list-size of uniformly random codes.

However, these results are more relevant for smaller alphabet sizes $q$, and approaching the alphabet-independent capacity of $p=1-R$ is less understood.
In this setting, uniformly random codes are, with high probability, list-decodable to capacity with optimal alphabet size $2^{O(1/\varepsilon)}$ \footnote{This follows from the list-decoding capacity theorem \cite{GRS22}. Over $q$-ary alphabets, the list-decoding capacity is given by $p=H_q^{-1}(1-R)$, which is larger than $1-R-\varepsilon$ when $q\ge 2^{\Omega(1/\varepsilon)}$.} and optimal list size $O(1/\varepsilon)$.\footnote{For codes over smaller alphabets, the list size $O(1/\varepsilon)$, where $\varepsilon$ is the gap to capacity, is believed to be optimal, but a proof is only known for large radius \cite{GV05}. However, for approaching the alphabet-independent capacity, the list size $O(1/\varepsilon)$ \emph{is} known to be optimal by the generalized Singleton bound \cite{ST20}.}
However, it was not known whether random linear codes (or, in general, more structured codes) could achieve similar parameters. In particular, both of the following questions were open (as far as we are aware).
\begin{itemize}
  \item Are rate $R$ random linear codes $(1-R-\varepsilon,O(1/\varepsilon))$-list-decodable with high probability? 
    Previously, this was not known for \emph{any} alphabet size $q$, even alphabet size growing with the length of the code.
    Previously, the best list size for random linear codes list-decodable to radius $p=1-R-\varepsilon$ was at least $2^{\Omega(1/\varepsilon)}$ \cite{GHK11,RudraW18}.\footnote{Prior works on list-decoding random linear codes were more relevant for $q$ an absolute constant such as $2,3,4,5,\dots$. \cite{GHK11} appears to give a list-size bound of $O(q^{O_R(1)}/\varepsilon)$, and \cite{RudraW18} appears to give a list size bound that is at least $q^{\log^2(1/\varepsilon)}$, and we need $q\ge 2^{\Omega(1/\varepsilon)}$.} 
  \item Do there exist \emph{any} linear codes (even non-constructively) over constant-sized (independent of $n$) alphabets that are $(1-R-\varepsilon,O(1/\varepsilon))$-list-decodable?
\end{itemize}
Using the same framework as the proof of Theorem~\ref{thm:rlc}, we answer both questions affirmatively.
We show that, with high probability, random linear codes approach the generalized Singleton bound, and thus capacity, with alphabet size close to the optimal.
\begin{theorem}
  For all $L\ge 1, \varepsilon\in (0,1)$, a random linear code over alphabet size $q\ge 2^{10L/\varepsilon}$ and $n$ sufficiently large is with high probability $(\frac{L}{L+1}(1-R-\varepsilon),L)$-average-radius-list-decodable.
  \label{thm:rlc}
\end{theorem}
By taking $L=O(1/\varepsilon)$, we see that random linear codes achieve capacity with optimal list size $O(1/\varepsilon)$ and near-optimal alphabet size $2^{O(1/\varepsilon^2)}$.
\begin{corollary}
  For all $\varepsilon > 0$, a random linear code over alphabet size $q\ge 2^{O(1/\varepsilon^2)}$ and $n$ sufficiently large is with high probability $(1-R-\varepsilon,O(1/\varepsilon))$-average-radius-list-decodable.
  \label{cor:rlc}
\end{corollary}

As the proof of Theorem~\ref{thm:rlc} is very similar to the proof of Theorem~\ref{thm:main}, we focus most of the paper on Theorem~\ref{thm:main} for brevity and clarity of presentation in Section~\ref{sec:prelim} and Section~\ref{sec:proof}.
In Section~\ref{sec:rlc}, we show how the definitions and proof can be modified to work for random linear codes.

\paragraph{Alphabet size lower bounds.}
Above, we saw that random linear codes achieve list-decoding capacity with optimal list-size and linear alphabet size.
A natural question is determining the optimal constant in the alphabet size. We showed that $q\ge n\cdot 2^{O(1/\varepsilon^2)}$ suffices, and by the list-decoding capacity theorem \cite{Elias91} --- which requires $q\ge 2^{\Omega(1/\varepsilon)}$ --- we cannot have better than an exponential-type dependence on $1/\varepsilon$ for subconstant $\varepsilon < O(1/\log n)$. 

For approaching capacity with constant $\varepsilon$, Ben-Sasson, Kopparty, and Radhakrishnan~\cite{BKR10} showed that, for any $c\ge 1$, there exist full-length Reed--Solomon codes that are not list-decodable much beyond the Johnson bound with list-sizes $O(n^c)$. Thus in order to achieve list-decoding capacity, one needs $q>n$ in some cases.
However, while full-length Reed--Solomon codes could not achieve capacity, perhaps it was possible that Reed--Solomon codes over field size, say $q=2n$ or even $q=(1+\gamma)n$, could achieve capacity in all parameter settings.
We observe that, as a corollary of \cite{BKR10}, such a strong guarantee is not possible: for any $c>1$, there exist a constant rate $R=R(c)>0$ and infinitely many field sizes $q$ such that all Reed--Solomon codes of length $n\ge q/c$ and rate $R$ over $\mathbb{F}_q$ are not list-decodable to capacity $1-R$ with list size $n^c$. The proof is in Section~\ref{app:lb}.
\begin{proposition}
  Let $\delta=2^{-b}$ for some positive integer $b\ge 3$.
  There exists infinitely many $q$ such that any Reed--Solomon code of length $n\ge 4\delta^{0.99}q$ and rate $\delta$ is not $(1-2\delta,n^{\Omega(\log(1/\delta))})$-list-decodable.
  \label{pr:lb}
\end{proposition}

\paragraph{Follow-up works.}
This paper is based on two conference papers \cite{GZ23,AGL24b}. The paper \cite{GZ23} proved Theorem~\ref{thm:main} with a quadratic alphabet size, while \cite{AGL24b} improved the alphabet size to linear and extended the techniques to random linear codes. There have already been several follow-ups to these works.

Brakensiek, Dhar, Gopi, and Zhang \cite{BDGZ23} proved that Algebraic Geometry (AG) codes achieve list-decoding capacity over constant-sized alphabets by combining our techniques with a generalized GM-MDS theorem, which Brakensiek, Dhar, and Gopi proved in \cite{BDG24}.

Recently, Guo, Xing, Yuan, and Zhang \cite{guo24} initiated the study of ``higher-order MRD codes'' as counterparts to higher-order MDS codes in the rank metric. They proved that random Gabidulin codes are list-decodable to capacity in the rank metric.

\section{Preliminaries}
\label{sec:prelim}

\subsection{Basic Notation}
For positive integers $t$, let $[t]$ denote the set $\{1,2,\dots,t\}$.
The \emph{Hamming distance} $d(x,y)$ between two vectors $x,y\in\mathbb{F}_q^n$ is the number of indices $i$ where $x_i\neq y_i$.
For a finite field $\mathbb{F}_q$, we follow the standard notation that $\mathbb{F}_q[X_1,\dots,X_n]$ denotes the ring of multivariate polynomials with variables $X_1,\dots,X_n$ over $\mathbb{F}_q$, and $\mathbb{F}_q(X_1,\dots,X_n)$ denotes the field of fractions of the polynomial ring $\mathbb{F}_q[X_1,\dots,X_n]$.
By abuse of notation, we let $X_{\le i}$ or $X_{[i]}$ to denote the sequence $X_1,\dots,X_i$, and we let, for example, $X_{\le i}=\alpha_{\le i}$ to denote $X_1=\alpha_1,X_2=\alpha_2,\dots,X_i=\alpha_i$.
Given a matrix $M$ over the field of fractions $\mathbb{F}_q(X_1,\dots,X_n)$ whose entries are in $\mathbb{F}_q[X_1,\dots,X_n]$ and field elements $\alpha_1,\dots,\alpha_i\in\mathbb{F}_q$,  let $M(X_{\le i}=\alpha_{\le i})$ denote the matrix over $\mathbb{F}_q(X_{i+1},X_{i+2},\dots,X_n)$ obtained by setting $X_{\le i}=\alpha_{\le i}$ in $M$.

\subsection{Hypergraphs and Connectivity}
\label{ssec:hypergraph}
In this work, we maintain a hypergraph perspective of the list-decoding problem, which was introduced in \cite{GLSTW21}. We describe a bad list-decoding instance with a hypergraph where the $L+1$ bad codewords identify the vertices and the $n$ evaluation points identify the hyperedges (Definition~\ref{def:hypergraph}).
While prior works described a bad list-decoding instance by $L+1$ sets indicating the agreements of the codewords with the received word, this hypergraph perspective gives us several advantages:

\begin{enumerate}
  \item The constraints imposed by a bad list-decoding configuration yield a hypergraph that is \emph{weakly-partition-connected}. This is a natural notion of hypergraph connectivity, which is well-studied in combinatorics~\cite{FKK03, FKK03b, Kiraly03} and optimization~\cite{JMS03, FrankK09, Frank11, CX18}, and which generalizes a well-known notion ($k$-partition-connectivity) for graphs \cite{nash1961edge,tutte1961problem}.\footnote{The notion of weakly-partition-connected sits between two other well-studied notions: \emph{$k$-partition-connected} implies $k$-weakly-partition-connected implies \emph{$k$-edge-connected} \cite{Kiraly03}. Each of these three notions generalizes an analogous notion on graphs. On graphs, $k$-partition-connected and $k$-weakly-partition-connected are equivalent.} This connection allows us to have more compact notation, definitions, and lemma statements, and allows us to streamline some proofs.
  \item With the hypergraph perspective, we can give a new presentation of the results in \cite{BGM23} and more directly connect the list-decoding problem to the GM-MDS theorem \cite{DauSY14,Lovett18,YildizH19}, as the heavy-lifting in the combinatorics is done using known results on hypergraph orientations. This is done in Appendix~\ref{app:bgm}.
\end{enumerate}

A hypergraph $\mathcal{H}=(V,\mathcal{E})$ is given by a set of vertices $V$ and a set $\mathcal{E}$ of \emph{(hyper)edges}, which are (possibly) subsets of the vertices $V$.
In this work, all hypergraphs have \emph{labeled} edges, meaning we enumerate our edges $e_i$ by distinct indices $i$ from some set, typically $[n]$, in which case we may also think of $\mathcal{E}$ as a tuple $(e_1,\dots,e_n)$.
Throughout this paper, the vertex set $V$ is typically $[t]$ for some positive integer $t$.
The \emph{weight} of a hyperedge $e$ is $\wt(e)\defeq \max(0,|e|-1)$, and the \emph{weight} of a set of hyperedges $\mathcal{E}$ is simply $\wt(\mathcal{E})\defeq \sum_{e\in \mathcal{E}}{\wt(e)}$.

\tikzstyle{vertex} = [fill,shape=circle,node distance=50pt,scale=0.7,font=\tiny,label={[font=\tiny]}]
\tikzstyle{edge} = [opacity=1,fill opacity=0,line cap=round, line width=2pt]
\tikzstyle{hyperedge} = [fill,opacity=1,fill opacity=1,line cap=round, line join=round, line width=12pt]
\newcommand\examplehypergraph{
        \draw (0,0.3) circle (4.5);
        \foreach \i in {1,...,7}{
          \node[vertex,label=above:\(f\ind{\i}\)] (\i) at (360*\i/7:3) {};
        }
        \begin{pgfonlayer}{background}
          \draw[hyperedge,color=green!60!gray!60] (1.center) -- (2.center) -- (4.center)--(1.center);
          \node[] at (-0.5,1) {$e_{n-2}$};
        \draw[hyperedge,color=orange!90!white!80] (6.center) -- (5.center) --(6.center);
        \node[] at (0.75,-2.7) {$e_{n-1}$};
        \draw[hyperedge,color=magenta!90!white!70] (7.center) -- (3.01,0);
        \node[] at (3.5,-0.8) {$e_{n}$};
        \end{pgfonlayer}
}
\begin{figure}[t]
  \begin{center}
    \begin{tikzpicture}[scale=0.5]
        \examplehypergraph
        \node[align=left] at (16,0) {
          \colorbox{green!60!gray!60}{$e_{n-2}=\{1,2,4\}$ means $f\ind{1}(\alpha_{n-2}) = f\ind{2}(\alpha_{n-2}) = f\ind{4}(\alpha_{n-2}) = y_{n-2}$}\\
          \colorbox{orange!90!white!80}{$e_{n-1}=\{5,6\}$ means $f\ind{5}(\alpha_{n-1}) = f\ind{6}(\alpha_{n-1}) = y_{n-1}$}\\
        \colorbox{magenta!90!white!70}{$e_n=\{7\}$ means $f\ind{7}(\alpha_n) = y_n$}};
    \end{tikzpicture}
    \end{center}
    \caption{Example edges from an agreement hypergraph $\mathcal{H}=([7], (e_1,\dots,e_n))$ (Definition~\ref{def:hypergraph}) arising from a bad list-decoding configuration with polynomials $f\ind{1},\dots,f\ind{7}\in\mathbb{F}_q[X]$, received word $y\in\mathbb{F}_q^n$, and evaluation points $\alpha_1,\dots,\alpha_n$.}
    \label{fig:hypergraph}
  \end{figure}
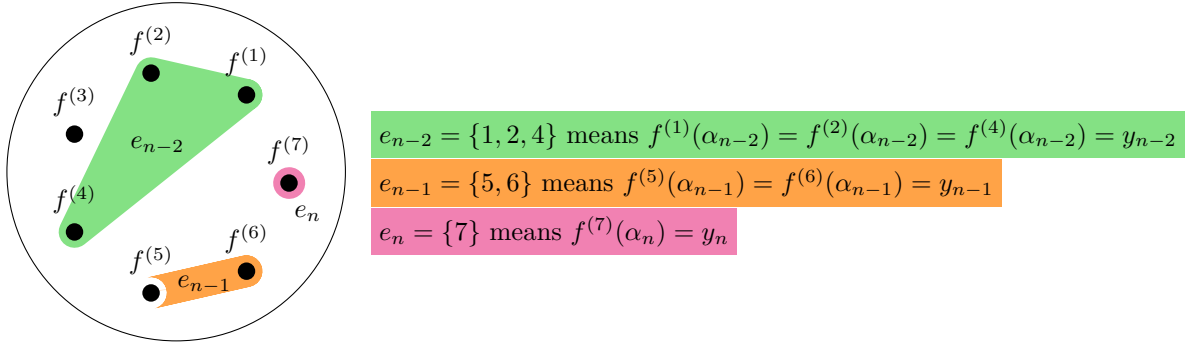

All hypergraphs that we will consider in this work are \emph{agreement hypergraphs} for a bad list-decoding configuration. See Figure~\ref{fig:hypergraph} for an illustration.
\begin{definition}[Agreement Hypergraph]
  Given vectors $y,c\ind{1},\dots,c\ind{t}\in\mathbb{F}_q^n$, the \emph{agreement hypergraph} has a vertex set $[t]$ and a tuple of $n$ hyperedges $(e_1,\dots,e_n)$ where $e_i \defeq \{j \in [t]: c^j_i = y_i\}$.
  \label{def:hypergraph}
\end{definition}
Agreement hypergraphs (or their subgraphs) enjoy a key property called weak-partition-connectivity. 

\begin{definition}[Weak-Partition-Connectivity]
  A hypergraph $\mathcal{H}=([t],\mathcal{E})$ is \emph{$k$-weakly-partition-connected} if, for every partition $\mathcal{P}$ of the set of vertices $[t]$,
  \begin{align}
  \sum_{e\in \mathcal{E}}{\max\{|\mathcal{P}(e)| - 1, 0\}} \ge k(|\mathcal{P}|-1)
    \label{eq:wpc}
  \end{align}
  where $|\mathcal{P}|$ is the number of parts of the partition, and $|\mathcal{P}(e)|$ is the number of parts of the partition that edge $e$ intersects.\footnote{We take $\max\{|\mathcal{P}(e)| - 1, 0\}$ rather than $|\mathcal{P}(e)| - 1$ because the hyperedges in agreement hypergraphs could be empty.}
\end{definition}

To give some intuition for weak-partition-connectivity, we state two of its combinatorial implications. First, if a hypergraph is $k$-weakly-partition-connected, then it is \emph{$k$-edge-connected} \cite{Kiraly03}, which, by the Hypergraph Menger's (Max-Flow-Min-Cut) theorem \cite[Theorem 1.11]{Kiraly03}, equivalently means that every pair of vertices has $k$ edge-disjoint (hyper)paths between them.\footnote{In general the converse is not true.}
Second, suppose we replace every hyperedge $e$ with an arbitrary spanning tree of its vertices (which we effectively do in Definition~\ref{def:rim}).
The resulting (non-hyper)graph is \emph{$k$-partition-connected},\footnote{In (non-hyper)graphs, $k$-partition-connectivity and $k$-weak-partition-connectivity are equivalent.} which, by the Nash-Williams-Tutte Tree-Packing theorem \cite{nash1961edge,tutte1961problem}, equivalently means there are $k$ edge-disjoint spanning trees (this connection was used in \cite{GLSTW21}).

The key reason we consider weak-partition-connectivity is that a bad list-decoding configuration yields a $k$-weakly-partition-connected agreement hypergraph.
\begin{lemma}[{Bad list gives $k$-weakly-partition-connected hypergraph. See also \cite[Lemma 7.4]{GLSTW21}}]
  Suppose that vectors $y,c\ind{1},\dots,c\ind{L+1}\in\mathbb{F}_q^n$ are such that the average Hamming distance from $y$ to $c\ind{1},\dots,c\ind{L+1}$ is at most $\frac{L}{L+1}(n-k)$. That is, $\sum_{j=1}^{L+1}d(y,c\ind{j}) \le L(n-k)$. Then, for some subset $J\subseteq[L+1]$ with $|J|\ge 2$, the agreement hypergraph of $(y,c\ind{j}:j\in  J)$ is $k$-weakly-partition-connected.
  \label{lem:hypergraph-1}
\end{lemma}
Lemma~\ref{lem:hypergraph-1} follows from the following result about weakly-partition-connected hypergraphs
\begin{lemma}
    Let $\mathcal{H}=(V,\mathcal{E})$ be a hypergraph with at least two vertices such that $\sum_{e\in\mathcal{E}} \wt(e)\ge k\cdot (|V|-1)$, where $k$ is a positive integer. Then there exists a subset $V'\subseteq V$ of at least two vertices such that the hypergraph $\mathcal{H}'=(V', \{e\cap V': e\in\mathcal{E}\})$ is $k$-weakly-partition-connected.
    \label{lem:hypergraph-0}
\end{lemma}
\begin{proof}
    Let $V'$ be an inclusion-minimal subset $V'\subseteq[L+1]$ with $|V'|\ge 2$ such that 
    \begin{align}
        \sum_{e\in\mathcal{E}}{^{} \wt(e\cap V')} \ge (|V'|-1)k.
        \label{eq:hypergraph-0-0}
    \end{align}
By assumption, $V'=[L+1]$ satisfies \eqref{eq:hypergraph-0-0}, so $V'$ exists (note that singleton subsets of $[L+1]$ satisfy \eqref{eq:hypergraph-0-0} with equality).
Let $\mathcal{H}=(V',\mathcal{E}')$ be the hypergraph with edge set $\mathcal{E}'=\{V'\cap e: e\in\mathcal{E}\}$.
By minimality of $V'$, for all nonempty $V''\subsetneq V'$, we have $\sum_{e\in\mathcal{E}'}{\wt(e\cap V'')} \le (|V''|-1)k$.
Now, consider a non-trivial partition $\mathcal{P}=P_1\sqcup\cdots\sqcup P_p$ of $V'$ where $P_i \neq V'$ for all $i \in [p]$ (as otherwise \eqref{eq:wpc} trivially follows). We have
\begin{align}
  \sum_{e\in\mathcal{E}'}{\max\{|\mathcal{P}(e)|-1, 0\}}
  &=\sum_{e\in\mathcal{E}'}\left(\wt(e) - \sum_{\ell=1}^p{\wt(e\cap P_{\ell})}\right) \nonumber \\
  &=\sum_{e\in\mathcal{E}'}{\wt(e)} - \sum_{\ell=1}^p{\sum_{e\in\mathcal{E}'}{\wt(e\cap P_{\ell})}} \nonumber \\
  &\ge (|V'|-1)k - \sum_{\ell=1}^p{(|P_{\ell}|-1)k} \nonumber \\
  &= (p-1)k \nonumber\\
  &= (|\mathcal{P}|-1)k.
\end{align}
This holds for all partitions $\mathcal{P}$ of $V'$, so $\mathcal{H}'$ is $k$-weakly-partition-connected.
\end{proof}
\begin{proof}[Proof of Lemma~\ref{lem:hypergraph-1}]
  Consider the agreement hypergraph $([L+1],\mathcal{E})$ of $y,(c\ind{1},\dots,c\ind{L+1})$.
The total edge weight is
\begin{align}
  \sum_{e\in \mathcal{E}}{\wt(e)}
  &\ge -n + \sum_{e\in \mathcal{E}}{|e|}
  = -n + \sum_{i=1}^{n}{\sum_{j=1}^{L+1}{\one[y_i = c\ind{j}_i]}}
   = -n + \sum_{j=1}^{L+1}{(n-d(y,c\ind{j}))}
   \ge Lk.
   \label{eq:hyper-1-1}
\end{align}
   By Lemma~\ref{lem:hypergraph-0}, there exists a subset $J\subseteq [L+1]$ of at least two vertices such that $\mathcal{H}'=(J, \{J\cap e: e\in\mathcal{E}\})$ --- which is exactly the agreement hypergraph of $(y,c\ind{j}:j\in J)$ --- is $k$-weakly-partition-connected.
\end{proof}

\begin{remark}
    The condition $|J|\ge 2$ is needed later so that the reduced intersection matrix (defined below) is not a $0\times 0$ matrix, in which case the matrix does not help establish list-decodability.
\end{remark}

\subsection{Reduced Intersection Matrices: Definition and Example}

We work with the reduced intersection matrix, which encodes all agreements from a bad list-decoding configuration into linear constraints on the message symbols (the polynomial coefficients). 
We point out that previous works \cite{ST20, GLSTW21, BGM23} considered a related matrix called the (non-reduced) \emph{intersection matrix}, and our proof would work just as well with this matrix.

\begin{definition}[Reduced intersection matrix]
  The \emph{reduced intersection matrix} $\mathsf{RIM}_{k,q,\mathcal{H}}$ associated with a prime power $q$, degree $k$, and a hypergraph $\mathcal{H}=([t],(e_1,\dots,e_n))$ is a $\wt(\mathcal{E}) \times (t-1)k$ matrix over the field of fractions $\mathbb{F}_q(X_1,\dots,X_n)$. It is constructed as follows.
  For each hyperedge $e_i$ with vertices $j_1<j_2<\dots<j_{|e_i|}$, we add $\wt(e_i) = |e_i|-1$ rows to $\mathsf{RIM}_{k,q,\mathcal{H}}$.
  For $u=2,\dots,|e_i|$, we add a row $r_{i,u}=(r^{(1)},\ldots,r^{(t-1)})$ of length $(t-1)k$ defined as follows:
  \begin{itemize}
    \item If $j = j_1$, then $r\ind{j} = [1, X_i, X_i^2,\dots,X_i^{k-1}]$ 
    \item If $j = j_u$ and $j_u \neq t$, then $r\ind{j} = -[1, X_i, X_i^2,\dots,X_i^{k-1}]$ 
    \item Otherwise, $r\ind{j} = 0^k$.
  \end{itemize}
  We typically omit $k$ and $q$ and write $\mathsf{RIM}_\mathcal{H}$ as $k$ and $q$ are typically understood.
    \label{def:rim}
\end{definition}
\begin{example}
  Recall the example edges of the agreement hypergraph $\mathcal{H}=([7],(e_1,\dots,e_n))$ in Figure~\ref{fig:hypergraph}.
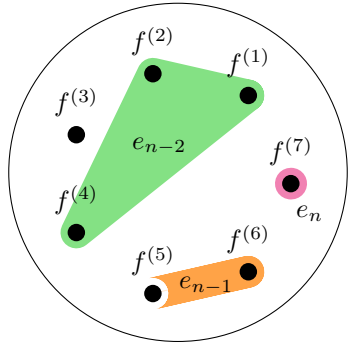
  \begin{center}
    \begin{tikzpicture}[scale=0.5]
      \examplehypergraph
    \end{tikzpicture}
  \end{center}
  The edges $e_{n-2},e_{n-1},e_n$ from $\mathcal{H}$ contribute the following length $(t-1)k$ rows to its reduced intersection matrix:
  \begin{align}
    \begin{bmatrix}
      V_{n-2} & -V_{n-2} & 0 & 0 & 0 & 0 \\
      V_{n-2} & 0& 0& -V_{n-2} & 0 & 0 \\
      0& 0& 0& 0 & V_{n-1} & -V_{n-1}  \\
    \end{bmatrix}
  \end{align}
  Here $V_{i} = [1,X_i,X_i^2,\dots,X_i^{k-1}]$ is a ``Vandermonde row'', and $0$ denotes the length-$k$ vector $[0,0,\dots,0]$.
  Note that each edge $e$ contributes $|e|-1$ rows to the agreement matrix, and in particular $e_n$ does not contribute any rows.
\end{example}

  The following lemma implies that, if every reduced intersection matrix arising from a possible bad list-decoding configuration has full column rank when $X_1=\alpha_1,\dots,X_n=\alpha_n$, the corresponding Reed--Solomon code is list-decodable. 
  
\begin{lemma}[RIM of agreement hypergraphs are not full column rank]
    Let $\mathcal{H}$ be an agreement hypergraph for $(y,c\ind{1},\dots,c\ind{t})$, where $c\ind{j}\in\mathbb{F}_q^n$ are codewords of $\RS_{n,k}(\alpha_1,\dots,\alpha_n)$, not all equal to each other.
    Then the reduced intersection matrix $\mathsf{RIM}_\mathcal{H}(X_{[n]}=\alpha_{[n]})$ does not have full column rank. \label{lem:eval}
\end{lemma}
\begin{proof}
    By definition, 
    \begin{align}
        \mathsf{RIM}_\mathcal{H}(X_{[n]}=\alpha_{[n]}) \cdot 
        \begin{bmatrix}
            f\ind{1}-f\ind{t}\\
            \vdots\\
            f\ind{t-1}-f\ind{t}
        \end{bmatrix}= 0
    \end{align}
    where $f\ind{1},\dots,f\ind{t}\in\mathbb{F}_q^{k}$ are the vectors of coefficients of the polynomials that generate the codewords $c\ind{1},\dots,c\ind{t}\in\mathbb{F}_q^n$.
    Since these vectors are not all equal to each other, $\mathsf{RIM}_\mathcal{H}(X_{[n]}=\alpha_{[n]})$ does not have full column rank.
\end{proof}
\begin{remark}[Symmetries of reduced intersection matrices]
From this definition, it should be clear that we can divide the variables $X_1,\dots,X_n$ into at most $2^L$ classes such that variables in the same class are \emph{exchangeable} with respect to the reduced intersection matrix $\mathsf{RIM}_\mathcal{H}$: if $e_i$ and $e_{i'}$ are the same hyperedge, then swapping $X_i$ and $X_{i'}$ yields the same reduced intersection matrix (up to row permutations). 
This observation turns out to be crucial for bringing the alphabet size all the way down to linear; without it, we would get a quadratic alphabet size.
\label{rem:exchangable}
\end{remark}

\begin{remark}
The pairwise distinctness requirement in the definition of average-radius-list-decodability (see Section~\ref{subsec:results}) is nonetheless crucial in the proof of Theorem~\ref{thm:main}, despite the weaker requirement in Lemma~\ref{lem:eval}. That is because we will eventually apply Lemma~\ref{lem:eval} on the subcollection of codewords given from Lemma~\ref{lem:hypergraph-1}, which can potentially be arbitrary. The guarantee that this subcollection of codewords is not all equal to each other would then follow from pairwise distinctness of the codewords in the original list.
\end{remark}

\subsection{Reduced Intersection Matrices: Full Column Rank}

The following theorem shows that reduced intersection matrices of $k$-weakly-partition-connected hypergraphs are nonsingular when viewed as a matrix over $\mathbb{F}_q(X_1,\dots,X_n)$.
This was essentially conjectured by Shangguan and Tamo \cite{ST20} and essentially established by Brakensiek, Gopi, and Makam \cite{BGM23}, who conjectured and showed, respectively, nonsingularity of the (non-reduced) intersection matrix under similar conditions.
By the same union bound argument as in \cite[Theorem 5.8]{ST20}, Theorem~\ref{lem:int-mat-1} already implies list-decodability of Reed--Solomon codes up to the generalized Singleton bound over exponentially large field sizes, which is \cite[Theorem 1.5]{BGM23}. 
For completeness, and to demonstrate how the hypergraph perspective more directly connects the list-decoding problem to the GM-MDS theorem, we include a proof of Theorem~\ref{lem:int-mat-1} in Appendix~\ref{app:bgm}.
\begin{theorem}[Full column rank. Implicit from Theorem A.2 of \cite{BGM23}]
    Let $n$ and $k$ be positive integers and $\mathbb{F}_q$ be a finite field. 
    Let $\mathcal{H}$ be a $k$-weakly-partition-connected hypergraph with $n$ hyperedges and at least $2$ vertices.
    Then $\mathsf{RIM}_{\mathcal{H}}$ has full column rank over the field $\mathbb{F}_q(X_1,\cdots,X_n)$.
    \label{lem:int-mat-1}
\end{theorem}
\begin{remark}
  We note that, \cite{BGM23} assumes throughout their paper that the alphabet size $q$ is sufficiently large, but, as in Theorem~\ref{lem:int-mat-1}, this assumption is easily dropped:
    For any fixed field size $q$, take $Q$ to be a sufficiently large power of $q$.
    Then, by the ``$q$ sufficiently large'' version of Theorem~\ref{lem:int-mat-1}, matrix $\mathsf{RIM}_{Q,\mathcal{H}}$ has full column rank over the field $\mathbb{F}_Q(X_1,\dots,X_n)$.
    Hence, the determinant of some square full-rank submatrix of $\mathsf{RIM}_{Q,\mathcal{H}}$ is a nonzero polynomial in $\mathbb{F}_Q[X_1,\dots,X_n]$.
    The entries of $\mathsf{RIM}_{Q,\mathcal{H}}$ can all be viewed as polynomials over $\mathbb{F}_q$, so the corresponding full-rank submatrix of $\mathsf{RIM}_{q,\mathcal{H}}$ has a determinant that is a nonzero polynomial in $\mathbb{F}_q[X_1,\dots,X_n]$ --- symbolically, the determinants are the same polynomials, as $\mathbb{F}_q$ and $\mathbb{F}_Q$ have the same characteristic.
    Hence, the matrix $\mathsf{RIM}_{q,\mathcal{H}}$ has full column rank over the field $\mathbb{F}_q(X_1,\dots,X_n)$.
    \label{rem:int-mat-1}
\end{remark}

\subsection{Reduced Intersection Matrix: Row Deletions}

We consider row deletions from the reduced intersection matrix. 
The goal of this section is to establish Lemma~\ref{lem:del}, that the full-column-rank-ness of reduced intersection matrices are robust to row deletions.

\begin{definition}[Row deletion of reduced intersection matrix]
  Given a hypergraph $\mathcal{H}=([t],(e_1,\dots,e_n))$ and set $B \subseteq [n]$, define $\mathsf{RIM}_{\mathcal{H}}^B$ to be the submatrix of $\mathsf{RIM}_\mathcal{H}$ obtained by deleting all rows containing a variable $X_i$ with $i\in B$.
  \label{def:del}
\end{definition}

The next lemma roughly says that, given a reduced intersection matrix $\mathsf{RIM}_\mathcal{H}$ with some constant factor ``slack'' in the combinatorial constraints, we can omit a constant fraction of the rows without compromising the full-column-rank-ness of the matrix.
\begin{lemma}[Robustness to deletions]
  Let $\mathcal{H}=([t],\mathcal{E})$ be a $(k+\varepsilon n)$-weakly-partition-connected hypergraph with $t \ge 2$, where $\mathcal{E}=(e_1,\dots,e_n)$.
  For all sets $B\subseteq [n]$ with $|B|\le \varepsilon n$, we have that $\mathsf{RIM}_{\mathcal{H}}^B$ is nonempty and has full column rank. 
  \label{lem:del}
\end{lemma}
\begin{proof}
  By definition of the reduced intersection matrix $\mathsf{RIM}_\mathcal{H}$, the matrix with row deletions $\mathsf{RIM}_\mathcal{H}^B$ is the matrix $\mathsf{RIM}_{\mathcal{H}'}$, where $\mathcal{H}'=([t],\mathcal{E}')$ is the hypergraph obtained from $\mathcal{H}$  by deleting $e_i$ for $i\in B$.
  By Theorem~\ref{lem:int-mat-1}, it suffices to prove that $\mathcal{H}'$ is $k$-weakly-partition connected.
  Indeed, consider any partition $\mathcal{P}$ of $[t]$.
  We have
  \begin{align}
    \sum_{e\in\mathcal{E}'}{\max\{|\mathcal{P}(e)|-1, 0\}}
    &= \sum_{i\in[n]}{\max\{|\mathcal{P}(e_i)|-1, 0\}}  - \sum_{i\in B}{\max\{|\mathcal{P}(e_i)|-1, 0\}} \nonumber\\
    &\ge (k+\varepsilon n)\cdot(|\mathcal{P}|-1)  - |B|\cdot (|\mathcal{P}|-1) 
    = k\cdot(|\mathcal{P}|-1), 
  \end{align}
  as desired.
  The first inequality holds because $\mathcal{H}$ is $(k+\varepsilon n)$-weakly-partition-connected, and, trivially, any edge $e_i$ touches at most $|\mathcal{P}|$ parts of $\mathcal{P}$.
\end{proof}

\section{Proof of List-Decodability with Linear-Sized Alphabets}
\label{sec:proof}

\subsection{Overview of the Proof}
\label{sec:proof-overview}

\tikzstyle{startstop} = [rectangle , text width=3cm , rounded corners, minimum width=3cm, minimum height=1cm, text centered, draw=black, fill=red!30]
\tikzstyle{arrow} = [thick, >=stealth]
\begin{figure*}
\label{diagram}
\centering
\begin{tikzpicture}[node distance=2cm]
\footnotesize

\node (1)[startstop, yshift=0cm, fill=orange!20] at (0,0) {\underline{Lemma~\ref{lem:eval}} \\ Bad list-decoding configuration has $(k+\varepsilon n)$-w.p.c agreement hypergraph};
\node (2)[startstop, fill=orange!20] at (4,0) {\underline{Lemma~\ref{lem:hypergraph-1}}\\ \textsf{RIM}s for agreement hypergraphs do not have full column rank};
\node (3)[startstop, fill=blue!50!green!20] at (8,0) {\underline{Lemma~\ref{lem:main}}\\ \textsf{RIM}s for $(k+\varepsilon n)$-w.p.c hypergraphs have full column rank w.h.p.};
\node (4)[startstop, fill=green!20] at (8,-3.5) {\underline{Theorem~\ref{thm:main}}\\ RS code list-decodable w.h.p.};
\node (x)[coordinate,above=1.5 of 4] {};
\draw [arrow] (2) |- (x);
\draw [arrow] (1) |- (x);
\draw [arrow] (3) |- (x);
\draw [arrow,->] (x) --node[anchor=east,minimum width=3cm,align=center]{Union bound over possible \\ agreement hypergraphs} (4);

\node (5)[startstop, fill=blue!50!green!20] at (0,3.5) {\underline{Lemma~\ref{lem:cert-2}}\\ If \textsf{RIM} not full column rank, it admits a certificate.};

\node (6)[startstop, fill=blue!50!green!20] at (4,3.5) {\underline{Corollary~\ref{cor:count}}\\ Number of possible certificates is small.};

\node (7)[startstop, fill=blue!50!green!20] at (8,3.5) {\underline{Corollary~\ref{cor:product}}\\ The probability of any one certificate is very small};

\node (y)[coordinate,above=1.5 of 3] {};
\draw [arrow] (5) |- (y);
\draw [arrow] (6) |- (y);
\draw [arrow] (7) |- (y);
\draw [arrow,->] (y) --node[anchor=east,minimum width=3cm,align=center]{Union bound over \\ possible certificates} (3);

\node [left of =5,color=blue!50!black,anchor=east]
{\begin{minipage}{3cm} \begin{center} Properties of \texttt{GetCertificate}, which generates certificates for non-full-rank \textsf{RIM}s.
\end{center}\end{minipage}};
\end{tikzpicture}

\caption{
A roadmap of our proof. The orange boxes are preliminaries, and the blue-green boxes are the meat of the proof address in Section~\ref{sec:proof}. All probabilities are over the random choice of evaluation points $\alpha_1,\dots,\alpha_n$ for our Reed--Solomon code.
}\label{fig:roadmap}
\end{figure*}
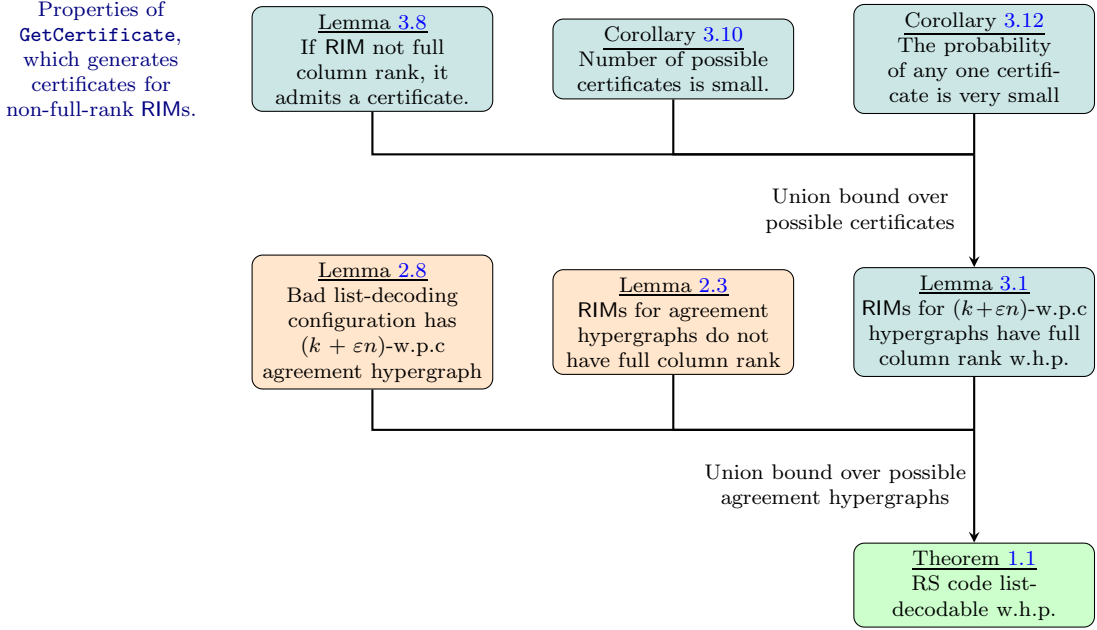

To explain our ideas, we first consider why previous results in \cite{ST20, GLSTW21, BGM23} require an exponentially large alphabet. 
By Lemma~\ref{lem:hypergraph-1} and Lemma~\ref{lem:eval}, every bad list-decoding configuration admits a weakly-partition-connected agreement hypergraph whose reduced intersection matrix does not have full column rank. 
Thus, to prove list-decodability, it suffices to show that, with high probability, every such reduced intersection matrix has full column rank.
Because these matrices have full column rank as symbolic matrices (Theorem~\ref{lem:int-mat-1}, see also \cite{BGM23}), they have full column rank under random evaluations with probability at least $1-\frac{Lk^2}{q}$ by the Schwarz-Zippel lemma (take the determinant of any full rank square submatrix).
Then, if $q$ is sufficiently large, by a union bound over all $2^{(L+1)n}$ possible reduced intersection matrices, all reduced intersection matrices have full column rank with high probability, so our code is list-decodable with high probability.
This union bound requires $q$ to be at least exponential in $n$.

Our key high-level insight is to use the ``slackness'' in the coding parameters to amplify the probability that the reduced intersection matrix fails to be full rank from $O_{L,k,n}(1/q)$ to $1/q^{\Omega(n)}$; that way, a smaller alphabet size $q$ suffices.
To see this, consider the toy problem of independently picking $m$ random row vectors $v_1,\cdots,v_m\in\mathbb{F}_q^n$ to form an $m\times n$ matrix $M$, which we want to have full column rank. If we choose $m=n$, which is the optimal choice of $m$, then the probability that M has full column rank is bounded by a function that is $\Theta(1/q)$, and this happens only if each $v_i$ is not in the span of $v_1,\dots, v_{i-1}$. However, suppose we choose $m = (1 + \varepsilon)n$ for some small $\varepsilon > 0$. In this case, we could afford $\varepsilon n$ ``faulty'' vectors $v_i$ , i.e., $v_i$ may be in the span of previous vectors, in which case we just skip it and consider the next vector. The probability that the matrix $M$ has full column rank is then exponentially small, $1/q^{\Omega(\varepsilon n)}$.
In the same way, the slackness in our coding parameters means the reduced intersection matrix has roughly $1+\varepsilon$ times as many rows as columns, so (with some additional ideas) we can similarly amplify the probability that the reduced intersection matrix fails to be full rank.

Lemma~\ref{lem:main} below captures this probability amplification. Our main result, Theorem~\ref{thm:main}, follows by applying Lemma~\ref{lem:hypergraph-1} and Lemma~\ref{lem:eval} with Lemma~\ref{lem:main}, and taking a union bound over all $\sum_{t=2}^{L+1}{2^{tn}}$ possible agreement hypergraphs.
\begin{lemma}
  Let $k$ be a positive integer and $\varepsilon>0$.
  For each $(k+\varepsilon n)$-weakly-partition-connected hypergraph $\mathcal{H}=([t],(e_1,\dots,e_n))$ with $t \ge 2$, we have, for $r=\floor{\varepsilon n/2}$,
  \begin{align}
    \Pr_{\alpha_1,\dots,\alpha_n \sim \mathbb{F}_q\text{ distinct}} \left[ \mathsf{RIM}_{\mathcal{H}}(X_{[n]}=\alpha_{[n]})\text{ does not have full column rank} \right] \le \binom{n}{r}2^{tr}\cdot \left(\frac{(t-1)k}{q-n}\right)^r \ .
    \label{eq:main-lem-1}
  \end{align}
  \label{lem:main}
\end{lemma}

At the highest level, our proof of Lemma~\ref{lem:main} is a union bound over ``not-full-column-rank certificates.''
For every sequence of evaluation points $(\alpha_1,\dots,\alpha_n) \in \mathbb{F}_q^n$ for which $\mathsf{RIM}_{\mathcal{H}}$ does not have full column rank, we show that there is a \emph{certificate} $(i_1,\dots,i_r) \in [n]^r$ of distinct indices in $[n]$ (Lemma~\ref{lem:cert-2}), which intuitively ``attests'' to the failure of the matrix $\mathsf{RIM}_\mathcal{H}$ to be full column rank.
We then show that, for any certificate $(i_1,\dots,i_r)$, the probability that $(\alpha_1,\dots,\alpha_n)$ has certificate $(i_1,\dots,i_r)$ is exponentially small. (More precisely, it will at most be $(\frac{(t-1)k}{q-n})^r$. See Corollary~\ref{cor:product}).
We then show that there are not too many certificates (Corollary~\ref{cor:count}), and then union bound over the number of possible certificates to obtain the desired result (Lemma~\ref{lem:main}).

Given an evaluation sequence $(\alpha_1,\dots,\alpha_n)$, we deterministically generate a certificate $(i_1,\dots,i_r)$ as described in Algorithm~\ref{alg:certificate}. Alongside, we also produce a sequence of $(t-1)k \times (t-1)k$ submatrices $M_1,\dots,M_r$ of $\mathsf{RIM}_\mathcal{H}$, as specified in Algorithm~\ref{alg:matrix}, with each matrix entirely determined by the indices $i_1,\dots,i_r$.
The first matrix, $M_1$, is a canonical full-rank submatrix of $\mathsf{RIM}_\mathcal{H}(X_{[n]})$. More generally, for each $j$, the matrix $M_j$ is a full-rank submatrix of $\mathsf{RIM}_\mathcal{H}(X_{[n]})$ and is a deterministic function of $i_1,\dots,i_{j-1}$. 
Additionally, each index $i_j$ depends only on $M_j$ and the evaluation points $\alpha_1,\dots,\alpha_n$.
In this way, we sequentially select $M_1,i_1,M_2,i_2,\dots,$ in that order.

To explain the choice of $i_j$, we note that evaluating $X_{[n]} = \alpha_{[n]}$ forces $\mathsf{RIM}_\mathcal{H}$ to not be full rank, then all of its $(t-1)k \times (t-1)k$ submatrices are not full rank. 
Thus if we sequentially ``reveal" $X_1=\alpha_1, X_2=\alpha_2,\dots$, then at some point, $M_j$ becomes not-full-rank.
We define $i_j$ to be the index such that setting $X_{i_j}=\alpha_{i_j}$ makes $M_j$ not-full-rank.
Conditioned on $M_j$ being full rank with $X_1=\alpha_1,\dots,X_{i_j-1}=\alpha_{i_j-1}$, the probability that $M_j$ becomes not-full-rank when setting $X_{i_j}=\alpha_{i_j}$ is at most $\frac{(t-1)k}{q-n}$: $\alpha_{i_j}$ is uniformly random over at least $q-n$ field elements, and the degree of $X_{i_j}$ in the determinant of $M_j$ is at most $(t-1)k$ (and the determinant is nonzero by definition).
It is possible to run the conditional probabilities in the correct order to conclude that the probability that a particular certificate $i_1,\dots,i_r$ is generated is at most $(\frac{(t-1)k}{q-n})^r$.

At this point, we can already obtain a quadratic alphabet size by proving a weaker version of \eqref{eq:main-lem-1}.
The above argument works if we choose $M_j$ to be any canonical, say the lexicographically smallest, full rank square submatrix of $\mathsf{RIM}_\mathcal{H}^{\{i_1,\dots,i_{j-1}\}}$.
By a union bound over all certificates---of which there are, naively, at most $n^r$---the probability that $\textsf{RIM}_\mathcal{H}(X_{[n]}=\alpha_{[n]})$ is not full rank is at most $n^r\cdot (\frac{(t-1)k}{q-n})^r$, which is exponentially small for $q=\Theta(n^2)$, so we get our list-decoding result with alphabet size $q=\Theta(n^2)$.
To improve the alphabet size to quadratic, we improve the bound on the number of certificates to $\binom{n}{r}2^{tr}$, which is much smaller --- $2^{\Theta(n)}$ rather than $n^{\Theta(n)}$ --- when $r=\Omega(n)$, the parameter regime of interest here.
Our savings in the number of certificates comes from leveraging that there are at most $2^t$ different ``types'' of hyperedges (see Remark~\ref{rem:exchangable}), and thus at most $2^t$ different types of variables $X_i$ in the reduced intersection matrix $\mathsf{RIM}_\mathcal{H}$.
With this observation in mind, we assume, without loss of generality, that the edges of $\mathcal{H}$ are ordered by their respective type (we can relabel the edges of $\mathcal{H}$, which effectively permutes the rows of $\mathsf{RIM}_\mathcal{H}$).

To reduce the number of certificates, we make a more deliberate choice of matrices $M_j$ that leverages the symmetries of $\mathsf{RIM}_{\mathcal{H}}$ (Remark~\ref{rem:exchangable}) and forces the certificates to have $O_t(1)$ increasing sequences $i_1<i_2<i_3<\cdots$. First, we ensure that we can keep a ``bank'' of $\Omega_t(r)$ unused variables of each of the $O_t(1)$ types. Then, starting with a full column rank submatrix $M$ of $\mathsf{RIM}_{\mathcal{H}}$ devoid of all variables in the ``bank,'' we start sequentially applying the evaluations $X_1 = \alpha_1, X_2 = \alpha_2, \ldots$. Whenever $M(X_{\le i_1} = \alpha_{\le i_1})$ turns singular, we find that the evaluation $X_{i_1} = \alpha_{i_1}$ is what 'caused' it to become singular. We then go to the ``bank'' to find a variable $X_{i_1'}$ of the same type as $X_{i_1}$ and ``re-indeterminate'' $M$ by replacing all instances of $X_{i_1}$ in $M$ with $X_{i_1'}$. That way, we ensure that $M$ is, in a sense, ``reused.'' Furthermore, we ensure $i_1' > i_1$, so that the matrix $M(X_{\le i_1} = \alpha_{\le i_1})$ is now nonsingular, so we can keep going. Of course, if we end up reaching the end (i.e. $M(X_{[n]} = \alpha_{[n]})$ is full column rank), then in fact, $\mathsf{RIM}_{\mathcal{H}}(X_{[n]} = \alpha_{[n]})$ is full column rank, and so the evaluations $(\alpha_1, \ldots, \alpha_n)$ were `good' after all.

Otherwise, if the evaluations $(\alpha_1, \ldots, \alpha_n)$ were `bad', then the submatrix $M$ couldn't have reached the end, and that can only happen if some specific type was completely exhausted from the bank. However, given the size of our initial bank, that must have meant that $M$ must have been ``re-indeterminated'' at least $\Omega_t(r)$ times. When that happens, we collect the indices $i_1,\dots,i_\ell$ that we gathered from this round, remove them from $\mathsf{RIM}_\mathcal{H}$, and repeat the process again with a refreshed bank. Since we only need $r$ indices, we end up doing at most $O_t(1)$ rounds. Because each round yields a strictly increasing sequence of indices of length at least $\Omega_t(r)$, then we end up getting a certificate consisting of at most $O_t(1)$ strictly increasing runs of total length $r$, of which there are at most $\binom{n}{r}\cdot O_t(1)^r$ by simple counting.

To be more concrete, when we generate the submatrix $M = M_1$, we ensure that any variable appearing in $M_1$ has the same type as $\Omega_t(r)$ variables that are \emph{not} in $M_1$ (but still in $\mathsf{RIM}_\mathcal{H}$).
This creates a ``bank'' of variables of each type.
Then, if $X_{\le i_1}=\alpha_{\le i_1}$ was the evaluation that made $M_1$ singular, we can get $M_2$ by replacing all copies of $X_{i_1}$ with some $X_{i_1'}$ that is of the same type and in the ``bank.''
Since variables $i_1$ and $i_1'$ are of the same type, they have analogous rows in the reduced intersection matrix 
$\mathsf{RIM}_\mathcal{H}$, so this new matrix $M_2$ is still a submatrix of $\mathsf{RIM}_\mathcal{H}$. 
Therefore, we can pick up where we left off with $M_1$ but with $M_2$ instead. That is, $M_2$ will in fact be full rank when we apply the evaluations $X_{\le i_1} = \alpha_{\le i_1}$. Thus the next index $i_2$ on which $M_2$ turns singular will be strictly greater than $i_1$.
We then repeat the process in $M_2$, replacing $X_{i_2}$ with some $X_{i_2'}$ that is in the ``bank'' and of the same type, getting $M_3$, and so on.
We can continue this process for $\Omega_t(r)$ steps because of the size of the bank of each type, so we get an increasing run of length $\Omega_t(r)$ in our certificate.
After we run out of some type in our bank, we remove the used indices $i_1,\dots,i_\ell$ from $\mathsf{RIM}_\mathcal{H}$ and repeat the process again with a refreshed bank. This continues for $O_t(1)$ times only, as we only need $r$ indices in the end.

\subsection{Setup for Proof of Lemma~\ref{lem:main}}
\label{ssec:proof-def}
We devote most of the remainder of this section to formally proving Lemma~\ref{lem:main}, and conclude with the proof of Theorem~\ref{thm:main}.

\paragraph{Types.}
For a hypergraph $\mathcal{H}=([t],(e_1,\dots,e_n))$, the \emph{type} of an index $i$ (or, by abuse of notation, the type of the variable $X_i$, or the edge $e_i$) is simply the set $e_i\subseteq [t]$.
There are $2^t$ types, and by abuse of notation, we identify the types by the numbers $1,2,\dots,2^t$ in an arbitrary fixed order with a bijection $\tau:2^{[t]}\to [2^t]$, where $2^{[t]}$ denotes the power set of $[t]$.
We say a hypergraph is \emph{type-ordered} if the hyperedges $e_1,\dots,e_n$ are sorted according to their type: $\tau(e_1) \le \tau(e_2)\le\cdots\le \tau(e_n)$.
Since permuting the labels of the edges of $\mathcal{H}$ preserves the rank of $\mathsf{RIM}_{\mathcal{H}}$ (it merely permutes the rows of $\mathsf{RIM}_\mathcal{H}$), we can without loss of generality assume in Lemma~\ref{lem:main} that $\mathcal{H}$ is type-ordered. 

\paragraph{Global variables.}
Throughout the rest of the section, we fix a positive integer $k$, parameter $\varepsilon>0$, and  $\mathcal{H}=([t],(e_1,\dots,e_n))$, a type-ordered $(k+\varepsilon n)$-weakly-partition-connected hypergraph with $t \ge 2$. We also fix
\begin{align}
  r\defeq \left\lfloor \frac{\varepsilon n}{2} \right\rfloor.
\end{align}

\begin{algorithm}[t]
  \caption{$\mathtt{GetMatrixSequence}$}
  \label{alg:matrix}
  \KwIn{indices $i_1,\dots,i_{j-1}\in[n]$ for some $j\ge 1$.}
  \KwOut{$M_1,\dots,M_j$, which are $(t-1)k\times (t-1)k$ matrices over $\mathbb{F}_q(X_1,X_2,\dots,X_n)$.}
  $B\gets \emptyset$, $i_0\gets \perp$, $\ell_0\gets\perp$\;
  \For{$\ell=1,\dots,j$}{ \label{line:outer-loop}
    \tcp*[l]{$M_\ell$ depends only on $i_1,\dots,i_{\ell-1}$}
    \If{$\ell > 1$}{
      \tcp*[l]{Fetch new index from bank $B$}
      $\tau\gets$ the type of $i_{\ell-1}$\label{line:tau}\;
      $s\gets$ number of indices among $i_{\ell_0},i_{\ell_0+1},\dots,i_{\ell-1}$ that are type $\tau$ \label{line:s}\;
      $i_{\ell-1}'\gets$ the $s$-th smallest element of $B$ that has type $\tau$\;
      \If{$i_{\ell-1}'$ \textrm{is defined}}{
        $M_\ell\gets$ the matrix obtained from $M_{\ell-1}$ by replacing all copies of $X_{i_{\ell-1}}$ with $X_{i_{\ell-1}'}$
        \label{line:ml}
      }
    }
    \If{$M_\ell$ \textrm{not yet defined}}{
      \tcp*[l]{Refresh bank $B$}
      $B\gets \emptyset$\;
      \For{$\tau=1,\dots,2^t$}{
        $B\gets B \cup \{\text{largest $\floor{r/2^t}$ indices of type $\tau$ in $[n]\setminus\{i_1,\dots,i_{\ell-1}\}$}\}$
      (if there are less than $\floor{r/2^t}$ indices of type $\tau$, then $B$ contains all such indices) \label{line:B}\;
      }
      $M_\ell\gets $ lexicographically smallest nonsingular $(t-1)k\times (t-1)k$ submatrix of $\mathsf{RIM}_{\mathcal{H}}^{B\cup \{i_1,\dots,i_{\ell-1}\}}$ \label{line:ml2}\;
      $\ell_0\gets \ell$ \tcp{new refresh index} \label{line:refresh-index}\;
    }
  }
  \Return{$M_1,\dots,M_j$}
\end{algorithm}

\begin{algorithm}[t]
  \caption{$\mathtt{GetCertificate}$}
  \label{alg:certificate}
  \KwIn{Evaluation points $(\alpha_1,\dots,\alpha_n) \in \mathbb{F}_q^n$.}
  \KwOut{A ``certificate'' $(i_1,\dots,i_r) \in [n]^r$.}
  \For{$j=1,\dots,r$}{
    \tcp*[l]{$M_1,\dots,M_{j-1}$ stay the same, $M_j$ is now defined}
    $M_1,\dots,M_j = \mathtt{GetMatrixSequence}(i_1,\dots,i_{j-1})$\;
    $i_j \gets $ smallest index $i$ such that $M_j(X_{\le i}=\alpha_{\le i})$ is singular\label{line:ij}\;
    \If{$i_j$\text{ not defined}}{
      \Return{$\perp$}
    }
  }
  \Return{$(i_1,\dots,i_r)$}
\end{algorithm}

\subsection{\texorpdfstring{$\mathtt{GetCertificate}$}{GetCertificates} and \texorpdfstring{$\mathtt{GetMatrixSequence}$}{GetMatrixSequence}: Basic Properties}

As mentioned at the beginning of this section, we design an algorithm,  Algorithm~\ref{alg:certificate}, that attempts to generate a certificate $(i_1,\dots,i_r)\in[n]^r$ for evaluation points $\alpha_1,\dots,\alpha_n$.
It uses Algorithm~\ref{alg:matrix}, a helper function that generates the associated square submatrices $M_1,\dots,M_r$ of $\mathsf{RIM}_\mathcal{H}$.
Below, we establish some basic properties of these algorithms.

First, we establish that the matrices outputted by \texttt{GetMatrixSequence} are well-defined.

\begin{lemma}[Output is well-defined]
  For all sequence of indices $i_1,\dots,i_{j-1}$, if $M_1,\dots,M_j$ is the output of the function $\mathtt{GetMatrixSequence}(i_1,\dots,i_{j-1})$, then $M_1,\dots,M_j$ are well-defined.
  \label{lem:well-defined}
\end{lemma}
 \begin{proof}
    If $\ell$ is a refresh index, then we have $|B\cup\{i_1,\dots,i_{\ell-1}\}|< |B| + r \le 2r \le \varepsilon n$, so by Lemma~\ref{lem:del}, $\mathsf{RIM}_\mathcal{H}^{B\cup\{i_1,\dots,i_{\ell-1}\}}$ is nonempty and has full column rank. Thus $M_\ell$ exists in Line~\ref{line:ml2}.
    If $\ell$ is not a refresh index, $M_\ell$ is always well-defined by definition.
\end{proof}

Next, we observe that \texttt{GetMatrixSequence} is an ``online'' algorithm.
\begin{lemma}[Online]
  $\mathtt{GetMatrixSequence}$ is a deterministic function of $i_1,\dots,i_{j-1}$, and it computes $M_\ell$ ``online'', meaning $M_\ell$ depends only on $i_1,\dots,i_{\ell-1}$ for all $\ell=1,\dots,j$ (and $M_1$ is always the same matrix).
  In particular, $\mathtt{GetMatrixSequence}(i_1,\dots,i_{j-1})$ is a prefix of $\mathtt{GetMatrixSequence}(i_1,\dots,i_j)$.
  \label{lem:online}
\end{lemma}
\begin{proof}
    By definition and Lemma~\ref{lem:well-defined}.
\end{proof}

\begin{definition}[Refresh index]
  In $\mathtt{GetMatrixSequence}$, in the outer loop over $\ell$, we say a \emph{refresh index} is an index $\ell$ obtained at Line~\ref{line:refresh-index} (i.e. when $M_\ell$ is defined on Line~\ref{line:ml2}).
  For example, $\ell=1$ is a refresh index.
\end{definition}

Our first lemma shows that the new indices we are receiving from $\mathtt{GetMatrixSequence}$ are in fact new.

\begin{lemma}[New variable]
  In $\mathtt{GetMatrixSequence}$, in the outer loop iteration over $\ell$ at Line~\ref{line:outer-loop}, if we reach Line~\ref{line:ml} of $\mathtt{GetMatrixSequence}$, variable $X_{i_{\ell-1}'}$ does not appear in $M_{\ell_0},M_{\ell_0+1},\dots,M_{\ell-1}$, where $\ell_0$ is the largest refresh index less than $\ell$.
  \label{lem:cert-0}
\end{lemma}
\begin{proof}
   Let $B$ be the set defined in Line~\ref{line:B} at iteration $\ell_0$.
   In iterations $\ell'=\ell_0,\ell_0+1,\dots,\ell$, the set $B$ is the same, and $i_{\ell-1}'$ is in this set $B$ by definition.
   Thus, the variable $X_{i_{\ell-1}'}$ does not appear in $M_{\ell_0}$ by definition.
   For $\ell'=\ell_0,\ell_0+1,\dots,\ell$, the $(\tau,s)$ pairs generated at Line~\ref{line:tau} and Line~\ref{line:s} are pairwise distinct, so $X_{i_{\ell-1}'}$ is not added to $M_{\ell'}$ for $\ell'=\ell_0+1,\dots,\ell-1$ and thus is not in $M_{\ell_0},M_{\ell_0+1},\dots,M_{\ell-1}$.
\end{proof}

To show that the probability of a particular certificate $(i_1,\dots,i_r)$ is small (Lemma~\ref{lem:product}, Corollary~\ref{cor:product}), we crucially need that $i_1,\dots,i_r$ are pairwise distinct.
The next lemma proves that this is always the case.
\begin{lemma}[Distinct indices]
  For any sequence of evaluation points $(\alpha_1,\dots,\alpha_n) \in \mathbb{F}_q^n$, the output of $\mathtt{GetCertificate}(\alpha_1,\dots,\alpha_n)$ is a sequence $(i_1,\dots,i_r) \in [n]^r$ of pairwise distinct indices.
  \label{lem:distinct}
\end{lemma}
\begin{proof}
  By definition of $i_\ell$ at Line~\ref{line:ij} of $\mathtt{GetCertificate}$, variable $X_{i_\ell}$ must be in $M_\ell$, so suffices to show that $M_\ell$ never contains any variable $X_i$ for $i\in\{i_1,\dots,i_{\ell-1}\}$.
  We induct on $\ell$.
  If $\ell$ is a refresh index, this is true by definition.
  If not, let $\ell_0$ be the largest refresh index less than $\ell$.
  By induction, $i_1,\dots,i_{\ell-2}$ are not in $M_{\ell-1}$, so we just need to show $i_{\ell-1}'$ (the new index replacing $i_{\ell-1}$ in $M_{\ell}$ at Line~\ref{line:ml}) is not any of $i_1,\dots,i_{\ell-1}$.
  It is not any of $i_1,\dots,i_{\ell_0-1}$ because none of those indices are in $B$ by definition.
  It is not any of $i_{\ell'}$ for $\ell'=\ell_0,\dots,\ell-1$, because $X_{i_{\ell'}}$ is in $M_{\ell'}$, but $X_{i_{\ell-1}'}$ is not, by Lemma~\ref{lem:cert-0}.
\end{proof}

\subsection{Bad Evaluation Points Admit Certificates}

Here, we establish Lemma~\ref{lem:cert-2}, that if some evaluation points make $\mathsf{RIM}_\mathcal{H}$ not full-column-rank, then $\mathtt{GetCertificate}$ outputs a certificate.
To do so, we first justify our matrix constructions, showing that the matrices in $\mathtt{GetMatrixSequence}$ are in fact submatrices of $\mathsf{RIM}_\mathcal{H}$.

\begin{lemma}[$\mathtt{GetMatrixSequence}$ gives submatrices of $\mathsf{RIM}_\mathcal{H}$]
  For all sequence of indices $i_1,\dots,i_{j-1}$, if $M_1,\dots,M_j$ is the output of $\mathtt{GetMatrixSequence}(i_1,\dots,i_{j-1})$, then $M_1,\dots,M_j$ are $(t-1)k\times (t-1)k$ submatrices of $\mathsf{RIM}_\mathcal{H}$.
  \label{lem:cert-1}
\end{lemma}
\begin{proof}
  We proceed with induction on $\ell=1,\dots,j$.
  First, if $\ell$ is a refresh index, then $M_\ell$ is a submatrix of $\mathsf{RIM}_\mathcal{H}$ by definition.
  In particular, $M_1$ is a submatrix of $\mathsf{RIM}_\mathcal{H}$, so the base case holds.
  Now suppose $\ell$ is not a refresh index and $M_{\ell-1}$ is a submatrix of $\mathsf{RIM}_\mathcal{H}$.
  Matrix $M_\ell$ is defined by replacing all copies of $X_{i_{\ell-1}}$ with $X_{i_{\ell-1}'}$.
  To check that $M_\ell$ is a submatrix of $\mathsf{RIM}_\mathcal{H}$, it suffices to show that
  \begin{itemize}
    \item[(i)] for each row of $\mathsf{RIM}_\mathcal{H}$ containing $X_{i_{\ell-1}}$, replacing all copies of $X_{i_{\ell-1}}$ with $X_{i_{\ell-1}'}$ gives another row of $\mathsf{RIM}_\mathcal{H}$, and
    \item[(ii)] the variable $X_{i_{\ell-1}'}$ does not appear in $M_{\ell-1}$.
  \end{itemize}
  The first item follows from the fact that indices $i_{\ell-1}$ and $i_{\ell-1}'$ are of the same type, so (i) holds by definition of types and $\mathsf{RIM}_\mathcal{H}$ (see also Remark~\ref{rem:exchangable}).
  The second item is Lemma~\ref{lem:cert-0}.
  Thus, $M_\ell$ is a submatrix of $\mathsf{RIM}_\mathcal{H}$, completing the induction.
\end{proof}

We now show that any $n$-tuple of bad evaluation points admits a certificate. 

\begin{lemma}[Bad evaluations points admit certificates]
  If $(\alpha_1,\dots,\alpha_n) \in \mathbb{F}_q^n$ are evaluation points such that $\mathsf{RIM}_\mathcal{H}(X_{[n]}=\alpha_{[n]})$ does not have full column rank, $\mathtt{GetCertificate}(\alpha_1,\dots,\alpha_n)$ returns a certificate $(i_1,\dots,i_r) \in [n]^r$ (rather than $\perp$).
  \label{lem:cert-2}
\end{lemma}
\begin{proof}
  Suppose for contradiction that $\mathtt{GetCertificate}$ returns $\perp$ at iteration $j$ in the loop.
  Then there is no index $i$ such that $M_j(X_{\le i}=\alpha_{\le i})$ is singular, so in particular, $M_j(X_{[n]}=\alpha_{[n]})$ is nonsingular and thus has full column rank.
  By Lemma~\ref{lem:cert-1}, $M_j$ is a submatrix of $\mathsf{RIM}_\mathcal{H}$, so we conclude $\mathsf{RIM}_\mathcal{H}$ has full column rank.
\end{proof}

\subsection{Bounding the Number of Possible Certificates}

In this section, we upper bound the number of possible certificates.
The key step is to prove the following structural result about certificates.

\begin{lemma}[Certificate structure]
  Given a sequence of evaluation points $(\alpha_1,\dots,\alpha_n) \in \mathbb{F}_q^n$ such that $\mathsf{RIM}_\mathcal{H}(X_{[n]} = \alpha_{[n]})$ is not full column rank, the return value $(i_1,\dots,i_r) = \mathtt{GetCertificate}(\alpha_1,\dots,\alpha_n)$ satisfies $i_{j-1} < i_j$ for all but at most $2^t$ values $j=2,\dots,r$.
  \label{lem:count}
\end{lemma}
\begin{proof}
  Let $(i_1,\dots,i_r)$ be the return of $\mathtt{GetCertificate}$, and let $M_1,\dots,M_r$ be the associated matrix sequence.
  By Lemma~\ref{lem:online}, we have $M_1,\dots,M_j = \mathtt{GetMatrixSequence}(i_1,\dots,i_{j-1})$ for $j=1,\dots,r$.
  Recall an index $\ell\in[r]$ is a \emph{refresh index} if $M_\ell$ is defined on Line~\ref{line:ml2} rather than Line~\ref{line:ml}.
  The lemma follows from two claims: 
  \begin{enumerate}
    \item[(i)] If $\ell>1$ is not a refresh index, then $i_{\ell-1}<i_\ell$.
    \item[(ii)] Any two refresh indices differ by at least $r/2^t$.
  \end{enumerate}

  To see claim (i), let $\ell_0$ be the largest refresh index less than $\ell$.
  By definition of a refresh index, the set $B$ stays constant between when $M_{\ell_0}$ is defined and when $M_{\ell}$ is defined.
  From the definition of $i_j$ at Line~\ref{line:ij} in $\mathtt{GetCertificate}$, we know that
  \begin{itemize}
    \item For $i< i_{\ell-1}$ the matrix $M_{\ell-1}(X_{\le i}=\alpha_{\le i})$ is nonsingular. 
    \item The matrix $M_{\ell}(X_{\le i_\ell}=\alpha_{\le i_\ell})$ is singular. 
  \end{itemize}
  Suppose for contradiction that $i_\ell < i_{\ell-1}$. (Note that $i_{\ell-1} \neq i_{\ell}$ by Lemma~\ref{lem:distinct}.)
  We contradict the first item by showing, using the second item, that $M_{\ell-1}(X_{\le i_\ell}=\alpha_{\le i_\ell})$ is also singular.
  By the definition of $\mathtt{GetMatrixSequence}$, since $\ell$ is not a refresh index, $M_\ell$ is defined in Line~\ref{line:ml}.
  By construction of $B$ and $i_{\ell-1}'$, we know that $i_{\ell-1}' > i_{\ell-1} > i_\ell$.
  Thus, not only is $M_\ell$ obtained from $M_{\ell-1}$ by replacing all copies of $X_{i_{\ell-1}}$ with $X_{i_{\ell-1}'}$, but $M_{\ell}(X_{\le i_\ell}=\alpha_{\le i_\ell})$ is also obtained by replacing all copies of $X_{i_{\ell-1}}$ with $X_{i_{\ell-1}'}$ in $M_{\ell-1}(X_{\le i_\ell}=\alpha_{\le i_\ell})$ .
  Moreover, the variable $X_{i_{\ell-1}'}$ does not appear in $M_{\ell-1}$ by Lemma~\ref{lem:cert-0}. So we conclude that, as $M_{\ell}(X_{\le i_\ell}=\alpha_{\le i_\ell})$ is singular, so is $M_{\ell-1}(X_{\le i_\ell}=\alpha_{\le i_\ell})$.

  Now we show claim (ii). Suppose $\ell_0$ and $\ell_1$ are consecutive refresh indices. If a variable of type $\tau$ appears in the matrix $M_{\ell_0}$, there must be exactly $\floor{r/2^t}$ indices of type $\tau$ in $B$ (if there were fewer, then $B\cup\{i_1,\dots,i_{\ell-1}\}$ would contain all indices of type $\tau$, and the corresponding variables would not appear in $\mathsf{RIM}_{\mathcal{H}}^{B\cup\{i_1,\dots,i_{\ell-1}\}}$).
    Let $\tau$ be the type of index $i_{\ell_1-1}$.
    Since $\ell_1$ is a refresh index, the number of indices of type $\tau$ among $i_{\ell_0},i_{\ell_0+1},\dots,i_{\ell_1-1}$ must therefore be $\floor{r/2^t}+1$.
    In particular, this means $\ell_1 - \ell_0 \ge \floor{r/2^t} + 1 \ge r/2^t$, as desired.
\end{proof}
\begin{corollary}[Certificate count]
  The number of possible outputs of $\mathtt{GetCertificate}$ is at most $\binom{n}{r}2^{tr}$.
  \label{cor:count}
\end{corollary}
\begin{proof}
  The certificate consists of $r$ distinct indices of $[n]$ by Lemma~\ref{lem:distinct}. We can choose those in $\binom{n}{r}$ ways. These indices are distributed between at most $2^t$ increasing runs by Lemma~\ref{lem:count}. We can distribute these indices between the $2^t$ increasing runs in at most $(2^t)^r$ ways. 
\end{proof}

\subsection{Bounding the Probability of One Certificate}

The goal of this section is to establish Corollary~\ref{cor:product}, which states that the probability of obtaining a particular certificate is at most $\left(\frac{(t-1)k}{q-n}\right)^r$.

\begin{lemma}
  Let $i_1,\dots,i_{r}\in[n]$ be pairwise distinct indices, and $M_1,\dots,M_r$ be $(t-1)k\times (t-1)k$ submatrices of $\mathsf{RIM}_\mathcal{H}$. Over randomly chosen pairwise distinct evaluation points $\alpha_1,\dots\alpha_n \in \mathbb{F}_q$, define the following events for $j=1,\dots,r$:
  \begin{itemize}
      \item $E_j$ is the event that $M_j(X_{\le i}=\alpha_{\le i})$ is non-singular for all $i < i_j$.
      \item $F_j$ is the event that $M_j(X_{\le i_j}=\alpha_{\le i_j})$ is singular.
    \end{itemize}
    The probability that all the events hold is at most $\left(\frac{(t-1)k}{q-n}\right)^{r}$.
    \label{lem:product}
\end{lemma}
\begin{proof}
  Note that the set of evaluation points $\alpha_1,\dots,\alpha_n$ for which events $E_j$ and $F_j$ occur depends only on $M_j$ and $i_j$.
  Furthermore, each of the events $E_j$ and $F_j$ depends only on $M_i$, $i_j$, and the evaluation points.
  Thus, by relabeling the index $j$, we may assume without loss of generality that $i_1<i_2<\cdots<i_r$.
  We emphasize that we are \emph{not} assuming that the output of $\mathtt{GetCertificate}$ satisfies $i_1<\cdots<i_r$ (this is not true).
  We are instead just choosing how we ``reveal'' our events $E_j$ and $F_j$: starting with the smallest index in $i_1, \ldots , i_r$ and ending with the largest index in it.

  We have
  \begin{align}
    \Pr_{\alpha_{[n]}}\left[\bigwedge_{j=1}^{r} (E_j\wedge F_j)\right]
    &=\prod_{j=1}^{r} \Pr_{\alpha_{[n]}}\left[E_j\wedge F_j|E_1\wedge F_1\wedge\cdots\wedge E_{j-1}\wedge F_{j-1}\right] \nonumber\\
    &\le\prod_{j=1}^{r} \Pr_{\alpha_{[n]}}\left[F_j|E_1\wedge F_1\wedge\cdots\wedge E_{j-1}\wedge F_{j-1}\wedge E_j\right] 
    \label{eq:main-lem-2}
  \end{align}
  Note that $E_1\wedge F_1\wedge \cdots \wedge E_{j-1}\wedge F_{j-1}\wedge E_j$ depends only on $\alpha_1,\dots,\alpha_{i_j-1}$, and $F_j$ depends only on $\alpha_1,\dots,\alpha_{i_j}$.
  For any $\alpha_1,\dots,\alpha_{i_j-1}$ for which $E_1\wedge F_1\wedge \cdots \wedge E_{j-1}\wedge F_{j-1}\wedge E_j$ holds, we have that $M_j(X_{\le i_j-1}= \alpha_{\le i_j-1})$ is a $(t-1)k\times (t-1)k$ matrix in $\mathbb{F}_q(X_{i_j},X_{i_j+1},\dots,X_n)$ whose determinant is a nonzero polynomial of degree at most $(t-1)k$ in each variable (the determinant contains at most $t-1$ rows including $X_{i_j}$, each time with maximum degree $k-1$). In particular, at most $(t-1)k$ values of $\alpha_{i_j}$ can make the determinant zero since, viewing the determinant as a polynomial in variables $X_{i_j+1},\dots,X_n$ with coefficients in $\mathbb{F}_q[X_{i_j}]$, any single nonzero coefficient becomes zero on at most $(t-1)k$ values of $\alpha_{i_j}$.
  Conditioned on $\alpha_1,\dots,\alpha_{i_j-1}$, the field element $\alpha_{i_j}$ is uniformly random over $q-i_j+1\ge q-n$ elements.
  Thus, we have, for all $\alpha_1,\dots,\alpha_{i_j-1}$ such that $E_1\wedge F_1\wedge \cdots \wedge E_{j-1}\wedge F_{j-1}\wedge E_j$,
  \begin{align}
    \Pr_{\alpha_{i_j}}\left[F_j|\alpha_1,\dots,\alpha_{i_j-1}\right]\le \frac{(t-1)k}{q-n}.
    \label{eq:product-3}
  \end{align}
  Since $E_1\wedge F_1\wedge \cdots \wedge E_{j-1}\wedge F_{j-1}\wedge E_j$ depends only on $\alpha_{\le i_j-1}$ and $F_j$ depends only on $\alpha_{\le i_j}$, we have
  \begin{align}
    \Pr_{\alpha_{[n]}}\left[F_j|E_1\wedge F_1\wedge\cdots\wedge E_{j-1}\wedge F_{j-1}\wedge E_j\right] \le \frac{(t-1)k}{q-n}.
  \end{align}
  Combining with \eqref{eq:main-lem-2} gives the desired result.
\end{proof}

The key result for this section is a corollary of Lemma~\ref{lem:product}.
\begin{corollary}[Probability of one certficiate]
   For any sequence $i_1,\dots,i_{r}\in[n]$, over randomly chosen pairwise distinct evaluation points $\alpha_1 , \dots , \alpha_n$, we have
   \label{cor:product}
  \begin{align}
    \Pr\left[\mathtt{GetCertificate}(\alpha_1,\dots,\alpha_n) = (i_1,\dots,i_{r})\right] 
    \le \left(\frac{(t-1)k}{q-n}\right)^{r}.
  \end{align}
\end{corollary}
\begin{proof}
  By Lemma~\ref{lem:distinct}, we only need to consider pairwise distinct indices $i_1,\dots,i_r$, otherwise the probability is 0.
  Let $M_1,\dots,M_r=\mathtt{GetMatrixSequence}(i_1,\dots,i_r)$.
  By Lemma~\ref{lem:cert-1}, matrices $M_1,\dots,M_r$ are all submatrices of $\mathsf{RIM}_\mathcal{H}$.
  Thus, Lemma~\ref{lem:product} applies.
  Let $E_1,\dots,E_r, F_1,\dots,F_r$ be the events in Lemma~\ref{lem:product}.
  If $\mathtt{GetCertificate}(\alpha_1,\dots,\alpha_n)=(i_1,\dots,i_r)$, then the definition of $i_j$ in Line~\ref{line:ij} of $\mathtt{GetCertificate}$ implies that events $E_j$ and $F_j$ both occur.
  By Lemma~\ref{lem:product}, the probability that all $E_j$ and $F_j$ hold is at most $(\frac{(t-1)k}{q-n})^r$, hence the result.
\end{proof}

\subsection{Finishing the Proof of Lemma~\ref{lem:main}}

\begin{proof}[Proof of Lemma~\ref{lem:main}]
  Recall (Section~\ref{ssec:proof-def}) that we fixed $\mathcal{H}$ to be a type-ordered $(k+\varepsilon n)$-weakly-partition-connected hypergraph.
  By Lemma~\ref{lem:cert-2}, if the matrix $\mathsf{RIM}_\mathcal{H}(X_{[n]}=\alpha_{[n]})$ does not have full column rank, then $\mathtt{GetCertificate}(\alpha_1,\dots,\alpha_n)$ is some certificate $(i_1,\dots,i_r)$.
  By Corollary~\ref{cor:product}, the probability that $\mathtt{GetCertificate}(\alpha_1,\dots,\alpha_n) = (i_1,\dots,i_r)$ holds is at most $\left(\frac{(t-1)k}{q-n}\right)^r$.
  By Corollary~\ref{cor:count}, there are at most $\binom{n}{r}2^{tr}$ certificates.
  Taking a union bound over possible certificates gives the lemma.
\end{proof}

\subsection{Finishing the Proof of Theorem~\ref{thm:main}}

\begin{proof}[Proof of Theorem~\ref{thm:main}]
  By Lemma~\ref{lem:hypergraph-1}, if $\RS_{n,k}(\alpha_1,\dots,\alpha_n)$ is not $\left(\frac{L}{L+1}(1-R-\varepsilon),L\right)$ average-radius list-decodable, then there exists a vector $y$ and pairwise distinct codewords $c\ind{1},\dots,c\ind{t}$ with $t \ge 2$ such that the agreement hypergraph $\mathcal{H} = ([t], \mathcal{E})$ is $(R+\varepsilon)n = (k+\varepsilon n)$-weakly-partition-connected. By Lemma~\ref{lem:eval}, the matrix $\mathsf{RIM}_{\mathcal{H}}(X_{[n]}=\alpha_{[n]})$ is not full column rank. 
  Now, the number of possible agreement hypergraphs $\mathcal{H}$ is at most $\sum_{t=2}^{L+1} 2^{tn} \le 2^{(L+2)n}$.
  Thus by the union bound over possible agreement hypergraphs $\mathcal{H}$ with Lemma~\ref{lem:main}, we have, for $r=\floor{\frac{\varepsilon n}{2}}$,
  \begin{align}
    &\Pr_{\alpha_{[n]}}\left[ \RS_{n,k}(\alpha_1,\dots,\alpha_n) \text{ is not $\left(\frac{L}{L+1}(1-R-\varepsilon),L\right)$ list-decodable} \right] \nonumber\\
    &\le \Pr_{\alpha_{[n]}}\left[ \exists \text{ $(k+\varepsilon n)$-w.p.c. agreement hypergraph } \mathcal{H} \text{ such that } \mathsf{RIM}_{\mathcal{H}}(X_{[n]}=\alpha_{[n]}) \text{ is not full column rank}\right] \nonumber\\
    &\le 2^{(L+2)n} \max_{\text{$(k+\varepsilon n)$-w.p.c. }\mathcal{H}}\quad\Pr_{\alpha_{[n]}}\left[ \mathsf{RIM}_{\mathcal{H}}(X_{[n]}=\alpha_{[n]}) \text{ is not full column rank}\right] \nonumber\\
    &\le 2^{(L+2)n} \cdot \binom{n}{r} 2^{(L+1)r} \left( \frac{Lk}{q-n} \right)^r 
    \le \left( 2^{(L+2)n/r}\cdot \frac{e n}{r}\cdot 2^{L+1} \frac{Lk}{q-n} \right)^r 
    \le 2^{-Ln},
  \end{align}
  as desired.
  Here, we used that $q = n + k\cdot 2^{10L/\varepsilon}$.
\end{proof}

\section{Random Linear Codes}
\label{sec:rlc}

In this section, we discuss how to modify the proof of Theorem~\ref{thm:main} to give Theorem~\ref{thm:rlc}, list-decoding for random linear codes (RLCs).
Our proof follows the roadmap in Figure~\ref{fig:roadmap}.
The proof is identical up to a few minor modifications, which we state here for brevity.
Below, we state the same lemmas as in the proof for Reed--Solomon codes, adjusted for random linear codes, and we highlight the key differences in purple.
We expect that our framework could be applied even more generally to show that other families of random codes --- beyond randomly punctured Reed--Solomon codes and random linear codes --- achieve list-decoding capacity with small alphabet sizes, assuming such codes satisfy an appropriate GM-MDS theorem.

\subsection{Preliminaries: Notation and Definitions}
The generator matrix $G\in \mathbb{F}_q^{n\times  k}$ of a random linear code has independent uniformly random entries in $\mathbb{F}_q$.
To transfer the proof for list-decoding Reed--Solomon codes to list-decoding random linear codes, a key analogy is to think of the generator matrix as a $n\times k$ matrix of $nk$ distinct indeterminates $(X_{i,\ell})_{i\in[n],\ell\in[k]}$, evaluated at $nk$ independent and uniformly random field elements $(\alpha_{i,\ell})_{i\in[n],\ell\in[k]}$.
\begin{align}
  \mathcal{G} &\defeq 
  \begin{bmatrix}
    X_{1,1} & \cdots & X_{1,k} \\
    \vdots & \ddots & \vdots\\
    X_{n,1} & \cdots & X_{n,k} \\
  \end{bmatrix} \in \mathbb{F}_q(X_{1,1},\dots,X_{n,k})^{n\times k}, \nonumber\\
  G &\defeq \mathcal{G} 
  \vert_{X_{[n]\times [n]}=\alpha_{[n]\times [k]}} \nonumber\\
  \mathcal{G}_i &\defeq [X_{i,1},\dots,X_{i,k}] \text{ (the $i$th row of $\mathcal{G}$).}
\end{align}
We note that our randomly punctured Reed--Solomon code can also be viewed as an evaluation of $\mathcal{G}$, where $X_{i,\ell}$ is assigned $\alpha_i^{\ell-1}$ where $\alpha_1,\dots,\alpha_n$ are random distinct field elements over $\mathbb{F}$.
In this light, one might expect our framework can also apply, and indeed it does.

Accordingly, we use similar indexing shorthand, where the notation $X_{[a]\times [b]}$ represents the $a\cdot b$ indeterminates $X_{1,1},X_{1,2},\dots,X_{a,b}$, and similarly for field elements $\alpha_{[a]\times [b]}$.
For field elements $\alpha_{1,1},\dots,\alpha_{a,b}$, we write $X_{[a]\times [b]}=\alpha_{[a]\times [b]}$ to denote $X_{i,\ell}=\alpha_{i,\ell}$ for $1\le i\le a$ and $1\le b\le \ell$.

We again use the notion of an agreement hypergraph in Section~\ref{ssec:hypergraph}, and Lemma~\ref{lem:hypergraph-1} still holds.
For each agreement hypergraph $\mathcal{H}$, we consider more general reduced intersection matrix $\textsf{RIM}_{\mathcal{H},\mathcal{G}}$, where the $X_i$-Vandermonde-rows are instead the $i$-th row of $\mathcal{G}$.
More precisely,
\begin{definition}[Reduced intersection matrix, Random Linear Codes, analogous to Definition~\ref{def:rim}.]
  The \emph{reduced intersection matrix} $\mathsf{RIM}_{\mathcal{H},\mathcal{G}}$ associated with a hypergraph $\mathcal{H}=([t],(e_1,\dots,e_n))$ is a $\wt(\mathcal{E}) \times (t-1)k$ matrix over the field of fractions $\mathbb{F}_q(X_{1,1},\dots,X_{n,k})$.
  For each hyperedge $e_i$ with vertices $j_1<j_2<\dots<j_{|e_i|}$, we add $\wt(e_i) = |e_i|-1$ rows to $\mathsf{RIM}_{\mathcal{H},\mathcal{G}}$.
  For $u=2,\dots,|e_i|$, we add a row $r_{i,u}=(r^{(1)},\ldots,r^{(t-1)})$ of length $(t-1)k$ defined as follows:
  \begin{itemize}
    \item If $j = j_1$, then $r\ind{j} = \mathcal{G}_i = [X_{i,1}, X_{i,2}, X_{i,3},\dots,X_{i,k}]$
    \item If $j = j_u$ and $j_u \neq t$, then $r\ind{j}  = -\mathcal{G}_i = -[X_{i,1}, X_{i,2}, X_{i,3},\dots,X_{i,k}]$
    \item Otherwise, $r\ind{j} = 0^k$.
  \end{itemize}
    \label{def:rlc:rim}
\end{definition}

\subsection{Preliminaries: Properties of RLC Reduced Intersection Matrices}

We have similar preliminaries for reduced intersection matrices of random linear codes.

\begin{lemma}[RIM of agreement hypergraphs are not full column rank, analogous to Lemma~\ref{lem:eval}]
  Let $\mathcal{H}$ be an agreement hypergraph for $(y,c\ind{1},\dots,c\ind{t})$, where $c\ind{j}\in\mathbb{F}_q^n$ are distinct codewords of the code generated by $\mathcal{G}|_{X_{[n]\times [k]}=\alpha_{[n]\times [k]}}$.
    Then the reduced intersection matrix $\mathsf{RIM}_{\mathcal{H},\mathcal{G}}(X_{[n]\times [k]}=\alpha_{[n]\times [k]})$ does not have full column rank. \label{lem:rlc:eval}
\end{lemma}
\begin{proof}
  Analogous to the proof of Lemma~\ref{lem:eval}.
\end{proof}

\begin{lemma}[RIM have full column rank, analogous to Theorem~\ref{lem:int-mat-1}]
  Let $\mathcal{H}$ be a $k$-weakly-partition-connected hypergraph with $n$ hyperedges and at least $2$ vertices.
  Then $\mathsf{RIM}_{\mathcal{H},\mathcal{G}}$ has full column rank over the field $\mathbb{F}_q(X_{1,1},\dots,X_{n,k})$.
  \label{lem:rlc:int-mat-1}
\end{lemma}
\begin{proof}
We note that the Reed--Solomon code reduced intersection matrix $\textsf{RIM}_\mathcal{H}$ can be obtained from the random linear code reduced intersection matrix $\textsf{RIM}_{\mathcal{H},\mathcal{G}}$ by setting the indeterminates $X_{i,\ell} = X_i^{\ell-1}$, so Lemma~\ref{lem:rlc:int-mat-1} immediately follows from Theorem~\ref{lem:int-mat-1}.
We emphasize that, while Reed--Solomon codes require large alphabet sizes $q\ge \Omega(n)$, Theorem~\ref{lem:int-mat-1} still holds for constant alphabet sizes $q$ (see Remark~\ref{rem:int-mat-1}), so we can use it here.
\end{proof}

We remark that Lemma~\ref{lem:rlc:int-mat-1} can be proven directly by following the proof framework of Theorem~\ref{lem:int-mat-1} in Appendix~\ref{subapp:int-mat-1-proof}, but instead substitute the use of Theorem~\ref{thm:gm-mds} with an analogous GM-MDS theorem for Random Linear Codes, which can be found in Lemma 7 of~\cite{DSY15} (Lemma 7 of~\cite{DSY15} only implies Lemma~\ref{lem:rlc:int-mat-1} for $q$ to be sufficiently large, but again by Remark~\ref{rem:int-mat-1} the $q$ sufficiently large version of Lemma~\ref{lem:rlc:int-mat-1} implies the lemma for all $q$). That way, the proof of Theorem~\ref{thm:rlc} relies only on the proof framework of Theorem~\ref{thm:main} and not on any of its lemmas.

We again define row deletions for reduced intersection matrices.

\begin{definition}[Row deletions, analogous to Definition~\ref{def:del}]
  Given a hypergraph $\mathcal{H}=([t],(e_1,\dots,e_n))$ and set $B \subseteq [n]$, define $\mathsf{RIM}_{\mathcal{H},\mathcal{G}}^B$ to be the submatrix of $\mathsf{RIM}_{\mathcal{H},\mathcal{G}}$ obtained by deleting all rows containing the row $\mathcal{G}_i$ with $i\in B$.
\end{definition}
Now we show that, as for Reed--Solomon codes, the full-column-rankness of reduced intersection matrices is robust to deletions.
\begin{lemma}[Robustness to deletions, analogous to Lemma~\ref{lem:del}]
  Let $\mathcal{H}=([t],\mathcal{E})$ be a $(k+\varepsilon n)$-weakly-partition-connected hypergraph with $t \ge 2$.
  For all sets $B\subseteq [n]$ with $|B|\le \varepsilon n$, we have that $\mathsf{RIM}_{\mathcal{H},\mathcal{G}}^B$ is nonempty and has full column rank. 
\end{lemma}
\begin{proof}
  The proof is identical to Lemma~\ref{lem:del}, where we instead use the full column rankness of $\mathsf{RIM}_{\mathcal{H},\mathcal{G}}$ for $k$-weakly-partition-connected $\mathcal{H}$ (Lemma~\ref{lem:rlc:int-mat-1}) rather than the full column rankness of $\mathsf{RIM}_{\mathcal{H}}$ (Theorem~\ref{lem:int-mat-1}).
\end{proof}

\subsection{The Proof}
\begin{algorithm}[t]
  \caption{$\mathtt{GetMatrixSequenceRLC}$}
  \label{alg:matrix-0}
  \KwIn{indices $i_1,\dots,i_{j-1}\in[n]$ for some $j\ge 1$.}
  \KwOut{$M_1,\dots,M_j$, which are $(t-1)k\times (t-1)k$ matrices over $\mathbb{F}_q(X_{1,1},\dots,X_{n,k})$.}
  $B\gets \emptyset$, $i_0\gets \perp$, $\ell_0\gets\perp$\;
  \For{$\ell=1,\dots,j$}{
    \tcp*[l]{$M_\ell$ depends only on $i_1,\dots,i_{\ell-1}$}
    \If{$\ell > 1$}{
      \tcp*[l]{Fetch new index from bank $B$}
      $\tau\gets$ the type of $i_{\ell-1}$\;
      $s\gets$ number of indices among $i_{\ell_0},i_{\ell_0+1},\dots,i_{\ell-1}$ that are type $\tau$ \;
      $i_{\ell-1}'\gets$ the $s$-th smallest element of $B$ that has type $\tau$\;
      \If{$i_{\ell-1}'$ \textrm{is defined}}{
        $M_\ell\gets$ the matrix obtained from $M_{\ell-1}$ by replacing all copies of row $\mathcal{G}_{i_{\ell-1}}$ with $\mathcal{G}_{i_{\ell-1}'}$
      }
    }
    \If{$M_\ell$ \textrm{not yet defined}}{
      \tcp*[l]{Refresh bank $B$}
      $B\gets \emptyset$\;
      \For{$\tau=1,\dots,2^t$}{
        $B\gets B \cup \{\text{largest $\floor{r/2^t}$ indices of type $\tau$ in $[n]\setminus\{i_1,\dots,i_{\ell-1}\}$}\}$
      (if there are less than $\floor{r/2^t}$ indices of type $\tau$, then $B$ contains all such indices) \;
      }
      $M_\ell\gets $ lexicographically smallest nonsingular $(t-1)k\times (t-1)k$ submatrix of $\mathsf{RIM}_{\mathcal{H},\mathcal{G}}^{B\cup \{i_1,\dots,i_{\ell-1}\}}$ \;
      $\ell_0\gets \ell$ \tcp{new refresh index} \;
    }
  }
  \Return{$M_1,\dots,M_j$}
\end{algorithm}

\begin{algorithm}[t]
  \caption{$\mathtt{GetCertificateRLC}$}
  \label{alg:certificate-0}
  \KwIn{Generator matrix entries $\alpha_{1,1},\dots,\alpha_{n,k} \in \mathbb{F}_q$.}
  \KwOut{A ``certificate'' $(i_1,\dots,i_r) \in [n]^r$.}
  \For{$j=1,\dots,r$}{
    \tcp*[l]{$M_1,\dots,M_{j-1}$ stay the same, $M_j$ is now defined}
    $M_1,\dots,M_j = \mathtt{GetMatrixSequenceRLC}(i_1,\dots,i_{j-1})$\;
    $i_j \gets $ smallest index $i$ such that $M_j(X_{[i]\times [k]}=\alpha_{[i]\times [k]})$ is singular\;
    \If{$i_j$\text{ not defined}}{
      \Return{$\perp$}
    }
  }
  \Return{$(i_1,\dots,i_r)$}
\end{algorithm}

The proof of Theorem~\ref{thm:rlc} follows similarly to the proof of Theorem~\ref{thm:main}.
Our key lemma, analogous to Lemma~\ref{lem:main} is to show that reduced intersection matrices of weakly-partition-connected hypergraphs are full column rank with high probability.
\begin{lemma}[Analogous to Lemma~\ref{lem:main}]
  Let $k$ be a positive integer and $\varepsilon>0$.
  For each $(k+\varepsilon n)$-weakly-partition-connected hypergraph $\mathcal{H}=([t],(e_1,\dots,e_n))$ with $t \ge 2$, we have, for $r=\floor{\varepsilon n/2}$,
  \begin{align}
    \Pr_{\alpha_{[n]\times [k]}} \left[ \mathsf{RIM}_{\mathcal{H},\mathcal{G}}(X_{[n]\times [k]}=\alpha_{[n]\times [k]})\text{ does not have full column rank} \right] \le \binom{n}{r}2^{tr}\cdot \left(\frac{t-1}{q}\right)^r \ .
  \end{align}
  \label{lem:rlc:main}
\end{lemma}
We highlight that our probability bound here is better than the one in Lemma~\ref{lem:main} for Reed--Solomon codes. This is because (i) all indeterminates in our generator matrix (and thus, the reduced intersection matrix) appear with degree 1 (rather than degree up to $k-1$), and (ii) our indeterminates are assigned independently uniformly at random, rather than random \emph{distinct} values.
Thus, the probability of any particular square submatrix matrix being made singular with an assignment is at most $\frac{t-1}{q}$, rather than $\frac{(t-1)k}{q-n}$: item (i) improves the numerator from $(t-1)k$ to $t-1$, and item (ii) improves the denominator from $q-n$ to $q$.
This improved probability bound means we can use a smaller alphabet size for random linear codes than for Reed--Solomon codes.
Other than this difference, the rest of our proof follows analogously. We include some more details for completeness.

We start with the same setup in Section~\ref{ssec:proof-def}, defining types in the same way, and starting with a $(k+\varepsilon n)$-weakly-partition-connected hypergraph $\mathcal{H}$ that we assume without loss of generality is type-ordered.
We again fix
\begin{align}
  r\defeq \left\lfloor \frac{\varepsilon n}{2} \right\rfloor
\end{align}

To prove Lemma~\ref{lem:rlc:main}, we similarly find a certificate $(i_1,\dots,i_r)$ for each singular reduced intersection matrix. 
This certificate is generated by an analogous algorithm, \texttt{GetCertificateRLC}, which uses an analogous helper function \texttt{GetMatrixSequenceRLC}.
We show this certificate has the same three properties:
\begin{enumerate}
  \item A bad generator matrix, namely a generator matrix for which the reduced intersection matrix is not full column rank, must yield a certificate.
  \item There are few possible certificates.
  \item The probability that a random generator matrix yields a particular certificate is small.
\end{enumerate}
We generate the certificate in a similar way. This time, instead of sequentially revealing the evaluation points, we sequentially reveal rows of the generator matrix, and $i_1$ indicates.

The first item is captured in the following lemma.
\begin{lemma}[Bad generator matrix admits certificate, analogous to Lemma~\ref{lem:cert-2}]
  If $\alpha_{1,1},\dots,\alpha_{n,k} \in \mathbb{F}_q$ are entries for the generator matrix such that $\mathsf{RIM}_{\mathcal{H},\mathcal{G}}(X_{[n]\times [k]}=\alpha_{[n]\times [k]})$ does not have full column rank, $\mathtt{GetCertificateRLC}(\alpha_{1,1},\dots,\alpha_{n,k})$ returns a certificate $(i_1,\dots,i_r) \in [n]^r$ (rather than $\perp$).
\end{lemma}
\begin{proof}
  Analogous to the proof of Lemma~\ref{lem:cert-2}.
\end{proof}

Just as for Reed--Solomon codes, we obtain the same bound on the number of possible certificates.
\begin{lemma}[Analogous to Corollary~\ref{cor:count}]
  The number of possible outputs of $\mathtt{GetCertificateRLC}$ is at most $\binom{n}{r}2^{tr}$.
  \label{lem:rlc:count}
\end{lemma}
\begin{proof}
   Analogous to the proof of Corollary~\ref{cor:count}.
\end{proof}

Lastly, we obtain an upper bound on the probability of obtaining a particular certificate.

\begin{lemma}[Probability of one certificate, analogous to Corollary~\ref{cor:product}]
  For any sequence $i_1,\dots,i_{r}\in[n]$, over independent uniformly random $\alpha_{1,1},\dots,\alpha_{n,k}$, we have
  \begin{align}
    \Pr\left[\mathtt{GetCertificateRLC}(\alpha_{1,1},\dots,\alpha_{n,k}) = (i_1,\dots,i_{r})\right] 
    \le \left(\frac{t-1}{q}\right)^{r}.
  \end{align}
   \label{cor:rlc:product}
\end{lemma}
Lemma~\ref{cor:rlc:product} is slightly different from the analogous result for Reed--Solomon codes, Corollary~\ref{cor:product}, so we provide a little more justification here.
Similar to Corollary~\ref{cor:product}, Lemma~\ref{cor:rlc:product} follows from a lemma analogous to Lemma~\ref{lem:product}.
\begin{lemma}[Analogous to Lemma~\ref{lem:product}]
  Let $i_1,\dots,i_{r}\in[n]$ be pairwise distinct indices, and $M_1,\dots,M_r$ be $(t-1)k\times (t-1)k$ submatrices of $\mathsf{RIM}_{\mathcal{H},\mathcal{G}}$. Over random generator matrix entries $\alpha_{1,1},\dots\alpha_{n,k} \in \mathbb{F}_q$, define the following events for $j=1,\dots,r$:
\begin{itemize}
  \item $E_j$ is the event that $M_j(X_{[i]\times [k]}=\alpha_{[i]\times [k]})$ is non-singular for all $i < i_j$.
  \item $F_j$ is the event that $M_j(X_{[i_j]\times [k]}=\alpha_{[i_j]\times [k]})$ is singular.
  \end{itemize}
  The probability that all the events hold is at most $(\frac{t-1}{q})^{r}$.
  \label{lem:rlc:product}
\end{lemma}
\begin{proof}[Proof of Lemma~\ref{lem:rlc:product}]
  The proof is similar to the proof of Lemma~\ref{lem:product}.
  Lemma~\ref{lem:product} follows from combining Equation~\eqref{eq:product-3} with the appropriate conditional probabilities.
  This lemma follows the same approach.
  We again assume without loss of generality $i_1<i_2<\cdots,i_r$.

  Here, we want, analogous to Equation~\eqref{eq:product-3}, for all $\alpha_{[i_j-1]\times [k]}$ such that $E_1\wedge F_1\wedge\cdots\wedge E_{j-1}\wedge F_{j-1}\wedge E_j$,
  \begin{align}
    \Pr_{\alpha_{\{i_j\}\times [k]}}\left[F_j|\alpha_{[i_j-1]\times [k]}\right]\le \frac{t-1}{q}.
    \label{eq:rlc:product-3}
  \end{align}
  To see \eqref{eq:rlc:product-3}, consider the determinant of $M_j(X_{[i_j-1]\times [k]}=\alpha_{[i_j-1]\times [k]})$, a $(t-1)k\times (t-1)k$ matrix in $\mathbb{F}_q(X_{\{i_j,i_j+1,\dots,n\}\times [k]})$. 
  View the determinant of $M_j(X_{[i_j-1]\times [k]}=\alpha_{[i_j-1]\times [k]})$ as a polynomial in variables $X_{\{i_j+1,\dots,n\}\times [k]}$ with coefficients in $\mathbb{F}_q[X_{i_j,1},\dots,X_{i_j,k}]$.
  It is nonzero because we assume $E_j$ holds, so there is some coefficient of the form $f(X_{i_j,1},\dots,X_{i_j,k})$ that is nonzero.
  Since matrix $M_j$ has at most $t-1$ rows containing any variables among $X_{i_j,1},\dots,X_{i_j,k}$, each appearing with total degree 1, the total degree of $X_{i_j,1},\dots,X_{i_j,k}$ in the determinant of $M_j$ is at most $t-1$.
  Thus, the total degree of $f(X_{i_j,1},\dots,X_{i_j,k})$ is at most $t-1$.
  Hence, by the Schwartz--Zippel lemma, $f$ becomes zero with probability at most $\frac{t-1}{q}$ over random $\alpha_{i_j,1},\dots,\alpha_{i_j,k}$.
  Thus, the probability that $F_j$ holds is at most $\frac{t-1}{q}$, giving \eqref{eq:rlc:product-3}.

Combining conditional probabilities as in Lemma~\ref{lem:product} gives the result.
\end{proof}

\begin{proof}[Proof of Theorem~\ref{thm:rlc}]
By Lemma~\ref{lem:hypergraph-1}, if our random linear code generated by $G$ is not $\left(\frac{L}{L+1}(1-R-\varepsilon),L\right)$ average-radius list-decodable, then there exists a vector $y$ and codewords $c\ind{1},\dots,c\ind{t}$ with $t \ge 2$ such that the agreement hypergraph $\mathcal{H} = ([t], \mathcal{E})$ is $(R+\varepsilon)n = (k+\varepsilon n)$-weakly-partition-connected. By Lemma~\ref{lem:rlc:eval}, the matrix $\mathsf{RIM}_{\mathcal{H},\mathcal{G}}(X_{[n]\times [k]}=\alpha_{[n]\times [k]})$ is not full column rank. 
  Now, the number of possible agreement hypergraphs $\mathcal{H}$ is at most $\sum_{t=2}^{L+1} 2^{tn} \le 2^{(L+2)n}$.
  Thus by the union bound over possible agreement hypergraphs $\mathcal{H}$ with Lemma~\ref{lem:rlc:main}, we have, for $r= \floor{\frac{\varepsilon n}{2}}$,
  \begin{align}
    &\Pr_{\alpha_{[n]\times [k]}}\left[ \text{Code generated by }\mathcal{G}|_{X_{[n]\times [k]}=\alpha_{[n]\times [k]}} \text{ is not $\left(\frac{L}{L+1}(1-R-\varepsilon),L\right)$ list-decodable} \right] \nonumber\\
    &\le \Pr_{\alpha_{[n]\times [k]}}
    \left[
    \begin{aligned}
     &\exists \text{ $(k+\varepsilon n)$-w.p.c. agreement hypergraph } \mathcal{H} \text{ such that}\\ &\mathsf{RIM}_{\mathcal{H},\mathcal{G}}(X_{[n]\times [k]}=\alpha_{[n]\times [k]}) \text{ is not full column rank}
    \end{aligned}
    \right] \nonumber\\
    &\le 2^{(L+2)n} \max_{\text{$(k+\varepsilon n)$-w.p.c. }\mathcal{H}}\quad\Pr_{\alpha_{[n]\times [k]}}\left[ \mathsf{RIM}_{\mathcal{H},\mathcal{G}}(X_{[n]\times [k]}=\alpha_{[n]\times [k]}) \text{ is not full column rank}\right] \nonumber\\
    &\le 2^{(L+2)n} \cdot \binom{n}{r} 2^{(L+1)r} \left( \frac{L}{q} \right)^r 
    \le \left( 2^{(L+2)n/r}\cdot \frac{e n}{r}\cdot 2^{L+1} \cdot \frac{L}{q} \right)^r 
    \le 2^{-Ln},
    \label{eq:rlc-prob}
  \end{align}
  as desired.
  Here, we used that $q = 2^{10L/\varepsilon}$.
\end{proof}

\begin{remark}
  Our tighter bound on the number of certificates that leveraged the symmetries of the matrices $\textsf{RIM}$ is crucial to obtaining a near-optimal constant alphabet size for Theorem~\ref{thm:rlc}; we would only have obtained linear $O(n)$ alphabet size otherwise.
In \eqref{eq:rlc-prob}, our upper bound on the non-list-decodability probability is
\begin{align}
  2^{(L+2)n} \cdot \binom{n}{r} 2^{(L+1)r} \cdot \left( \frac{L}{q} \right)^r,
\end{align}
where $r=\varepsilon n/2$, where $\varepsilon>0$ is roughly the gap to capacity.
Here $\binom{n}{r}2^{(L+1)r}$ is the number of possible certificates.
If we had naively bounded the number of certificates by $n^r$, our bound on the non-list-decodability probability would then be
\begin{align}
  2^{(L+2)n} \cdot n^r \cdot \left( \frac{L}{q} \right)^r.
\end{align}
For this bound to be $o(1)$, we need to take $q \ge 2^{L/\varepsilon}\cdot n$, giving an alphabet size of $O(n)$.
This would still have been a new result, but leveraging the symmetries allowed us to achieve a near-optimal constant list size of $2^{O(L/\varepsilon)}$.
\end{remark}

\section{Alphabet Size Limitations}
\label{app:lb}

In this section, we establish Proposition~\ref{pr:lb}.
For positive integers $m$, view $\mathbb{F}_{2^m}$ as a vector space of dimension $m$ over base field $\mathbb{F}_2$.
For a set $S\subseteq \mathbb{F}_{2^m}$, let
\begin{align}
  P_S(X)\defeq \prod_{\alpha\in S}^{} (X-\alpha).
\end{align}
An \emph{affine subspace} is a set $L+\alpha=\{\alpha+\beta:\beta\in L\}$ for some subspace $L$ of $\mathbb{F}_{2^m}$.
\begin{lemma}[Proposition 3.2 of \cite{BKR10}]
  Let $L$ be a subspace of $\mathbb{F}_{2^m}$.
  Then $P_L$ has the form
  \begin{align}
    X^{2^{\dim L}} + \sum_{i=0}^{\dim L - 1} \alpha_i X^{2^i},
  \end{align}
  where $\alpha_i \in \mathbb{F}_{2^m}$.
  \label{lem:bkr-1}
\end{lemma}
As an immediate corollary, we have
\begin{lemma}
  Let $L$ be an affine subspace of $\mathbb{F}^{2^m}$.
  Then $P_L$ has the form
  \begin{align}
    X^{2^{\dim L}} + \sum_{i=0}^{\dim L - 1} \alpha_i X^{2^i} + \beta
  \end{align}
  for $\alpha_i,\beta\in \mathbb{F}_{2^m}$.
  \label{lem:bkr-2}
\end{lemma}
\begin{proof}
  Since $L$ is an affine subspace, there exists $\gamma$ such that $L - \gamma\defeq \{\alpha-\gamma:\alpha\in L\}$ is a subspace of $\mathbb{F}_{2^m}$.
  By Lemma~\ref{lem:bkr-1}, we have $P_{L-\gamma}$ is of the form
  \begin{align}
    X^{2^{\dim L}} + \sum_{i=0}^{\dim L - 1} \alpha_i X^{2^i}
  \end{align}
  for $\alpha_i\in \mathbb{F}_{2^m}$.
  In particular, $P_{L-\gamma}$ is $\mathbb{F}_2$-linear, so
  \begin{align}
    P_L(X) = P_{L-\gamma}(X+\gamma) = P_{L-\gamma}(X) + P_{L-\gamma}(\gamma).
  \end{align}
  Setting $\beta = P_{L-\gamma}(\gamma)$ gives the desired form for $P_L(X)$.
\end{proof}

\begin{lemma}[Analogous to Lemma 3.5 of \cite{BKR10}]
  Let $S$ be a subset of $\mathbb{F}_{2^m}$ of size $n$.
  Let $u$ and $v$ be integers such that $0\le u\le v\le m$.
  Then there is a family $\mathcal{L}$ of at least $2^{(u+1)m-v^2}$ affine subspaces of dimension $v$, such that each affine subspace $L\in\mathcal{L}$ satisfies $|L\cap S|\ge n/2^{m-v}$, and for any two affine subspaces $L,L'\in\mathcal{L}$, the difference $P_{L}-P_{L'}$ has degree at most $2^u$.
 \label{lem:bkr-3}
\end{lemma}
\begin{proof}
  For every subspace $L$ of dimension $v$, there exists $\beta_0,\dots,\beta_{2^{m-v}}$ such that the affine subspaces $L+\beta_i$ partition $\mathbb{F}_{2^m}$.
  By pigeonhole, there exists some $\beta_i$ such that $|(L+\beta_i)\cap S| \ge |S|/2^{m-v} = n/2^{m-v}$
  The number of subspaces of dimension $v$ is
  \begin{align}
    \frac{(2^m-1)(2^m-2)\cdots(2^m-2^{v-1})}{(2^v-1)(2^v-2)\cdots(2^v-2^{v-1})} \ge 2^{v(m-v)},
  \end{align}
  so there are at least $2^{v(m-v)}$ affine-subspaces $L$ with $|L\cap S|\ge n/2^{m-v}$.
  For all such affine-subspaces $L$, the polynomial $P_L(X)$ has the form $X^{2^v} + \sum_{i=0}^{v-1} \alpha_iX^{2^i}+\beta$ by Lemma~\ref{lem:bkr-2}.
  Among these affine-subspaces $L$, by the pigeonhole principle, for at least a fraction $2^{-m(v-u-1)}$ of these subspaces, their subspace polynomials $P_L(X)$ have the same $\alpha_i$ for $i=u+1,u+2,\dots,v$.
  Let $\mathcal{L}$ be this family of subspaces.
  The number of subspaces is at least $2^{v(m-v)}\times 2^{-m(v-u-1)}=2^{(u+1)m-v^2}$, so $\mathcal{L}$ is the desired family of affine subspaces.
\end{proof}

\begin{proof}[Proof of Proposition~\ref{pr:lb}]
  Let $\delta=2^{-r-1}$ as in the statement of Proposition~\ref{pr:lb}.
  Consider a Reed--Solomon code of length $n$ and rate $\delta$ over $\mathbb{F}_q$, where $q=2^m$ with $m$ sufficiently large.
  Let $S\subseteq \mathbb{F}_q$ be the set of $n$ evaluation points.
  Apple Lemma~\ref{lem:bkr-3} with $u=m-\ceil{1.99r}$ and $v=m-r$.
  This gives a family $\mathcal{L}$ of $2^{m(m-\ceil{1.99r}) - (m-r)^2} = 2^{2rm-\ceil{1.99r}m +r^2}\ge q^{\Omega(\log(1/\delta))}$ affine subspaces $L\le \mathbb{F}_{2^m}$ for which $|L\cap S|\ge n/2^{m-v} = 2\delta n$.
  Furthermore, for $L\in \mathcal{L}$, the subspace polynomials $P_L$ each have $2^v$ roots, and agree on all coefficients of degree larger than $2^u$. 
  Let $L_0$ be an arbitrary element of $\mathcal{L}$.
  Then the polynomials $\{P_{L_0}-P_{L}:L\in \mathcal{L}\}$ are each of degree at most $2^u = 2^{-\ceil{1.99r}}q\le 4\delta^{1.99}q\le\delta n$, and each agree with $P_{L_0}(X)$ on at least $|L\cap S|\ge 2\delta n$ values of $S$.
  Thus, our Reed--Solomon code is not $(1-2\delta, n^{\Omega(1/\delta)})$-list-decodable, as desired.
\end{proof}

\appendix 
\section*{Appendix}

\section{Alternate Presentation of \texorpdfstring{\cite{BGM23}}{[BGM23]}}
\label{app:bgm}
Here, we include alternate presentations of some ideas from \cite{BGM23}.
Algebraically, our presentation is the same, but the hypergraph perspective streamlines combinatorial aspects of their ideas.

\subsection{Preliminaries}
\paragraph{Dual of Reed--Solomon codes.}
It is well known that the \emph{dual} of a Reed--Solomon code is a \emph{generalized Reed--Solomon code}: Given positive integers $k \le n$ and evaluation points $\alpha_1,\dots,\alpha_n \in \mathbb{F}_q$, there exists nonzero $\beta_1,\dots,\beta_n\in\mathbb{F}_q$ such that the following matrix, called the \emph{parity-check matrix},
\begin{align}
  H = 
    \begin{bmatrix}
      \beta_1 & \beta_2 & \cdots & \beta_n \\
      \beta_1\alpha_1 & \beta_2\alpha_2 & \cdots & \beta_n\alpha_n \\
      \vdots &\vdots&&\vdots \\
      \beta_1\alpha_1^{n-k-1} & \beta_2\alpha_2^{n-k-1} & \cdots & \beta_n\alpha_n^{n-k-1} \\
    \end{bmatrix}
    \label{eq:rs-h}
\end{align}
satisfies $Hc = 0^{n-k}$ if and only if $c\in \RS_{n,k}(\alpha_1,\dots,\alpha_n)$.

\paragraph{Generic zero patterns.}
Following \cite{BGM23}, we leverage the GM-MDS theorem to establish the list-decodability of Reed--Solomon codes.
In this work, we more directly connect the list-decoding problem to the GM-MDS theorem using a hypergraph orientation lemma (introduced in the next section).
Here, we review generic zero patterns and the GM-MDS theorem.
To keep the meaning of the variable ``$k$'' consistent throughout the paper, we unconventionally state the definition of zero patterns and the GM-MDS theorem with $n-k$ rows instead of $k$ rows.
\begin{definition}
  Given positive integers $k \le n$, an $(n,n-k)$-\emph{generic zero pattern} (GZP) is a collection of sets $S_1, \ldots , S_{n-k} \subseteq [n]$ such that, for all $K \subseteq [n-k]$,
  \begin{align}
    \abs{\bigcap_{\ell\in K} S_\ell} \le n - k - |K|.
    \label{eq:gzp}
  \end{align}
\end{definition}

\paragraph{GM-MDS Theorem.}
As in \cite{BGM23}, we connect the list-decoding problem to the GM-MDS theorem. Here, we make the connection more directly.
\begin{theorem}[GM-MDS Theorem \cite{DauSY14, Lovett18,YildizH19}]
  Given $q\ge 2n-k-1$ and any generic zero pattern $S_1,\dots,S_{n-k} \subseteq [n]$, there exists pairwise distinct evaluation points $\alpha_1,\dots,\alpha_n \in \mathbb{F}_q$ and an invertible matrix $M\in \mathbb{F}_q^{(n-k)\times (n-k)}$ such that, if $H$ is the parity-check matrix for $\RS_{n,k}(\alpha_1,\dots,\alpha_n)$ (as in \eqref{eq:rs-h}), then $MH$ achieves zero pattern $S_1,\dots,S_{n-k}$, meaning that $(MH)_{\ell,i}=0$ whenever $i\in S_\ell$.
  \label{thm:gm-mds}
\end{theorem}
We note that the original GM-MDS theorem shows that the generator matrix of a (non-generalized) Reed Solomon code achieves any generic zero pattern. Here, we state that the generator matrix of a \emph{generalized} Reed--Solomon code achieves any generic zero pattern, which is an immediate corollary of the former result.

\subsection{Hypergraph Orientations}
Our new perspective of the tools from \cite{BGM23} leverages a well-known theorem about orienting weakly-partition-connected hypergraphs, stated below. This theorem is most explicitly stated in \cite{Frank11}, but it is implicit in \cite{Kiraly03,FKK03b}.

  A \emph{directed hyperedge} is a hyperedge with one vertex assigned as the \emph{head}. All the other vertices in the hyperedge are called \emph{tails}. 
  A \emph{directed hypergraph} consists of directed hyperedges.
  In a directed hypergraph, the \emph{in-degree} of a vertex $v$ is the number of edges for which $v$ is the head.
  A \emph{path} in a directed hypergraph is a sequence $v_1,e_1,v_2,e_2,\dots,v_{s-1},e_{s-1},v_s$ such that for all $\ell=1,\dots,s-1$, vertex $v_\ell$ is a tail of edge $e_\ell$ and vertex $v_{\ell+1}$ is the head of edge $e_\ell$.
  An \emph{orientation} of an (undirected) hypergraph is obtained by assigning a head to each hyperedge, making every hyperedge directed.
  \begin{theorem}[Theorems 9.4.13 and 15.4.4 of \cite{Frank11}]
    A hypergraph $\mathcal{H}$ is $k$-weakly-partition-connected if and only if it has an orientation such that, for some vertex $v$ (the ``root''), every other vertex $u$ has $k$ edge-disjoint paths to $v$.\footnote{In \cite[Theorems 9.4.13 and 15.4.4]{Frank11}, the property of having $k$ edge-disjoint paths to $v$ is called \emph{$(0,k)$-edge-connected}.}
  \label{thm:hypergraph}
  \end{theorem}
  We note that Theorem~\ref{thm:hypergraph} remains true if ``to $v$'' is replaced with ``from $v$'' and $k$-weakly-partition-connected is replaced with another hypergraph notion called $k$-partition-connected.
  The following corollary essentially captures (the hard direction of) \cite[Lemma~2.8]{BGM23}.
  \begin{corollary}
    Let $\mathcal{H}=([t],\mathcal{E})$ be a $k$-weakly-partition-connected hypergraph with $n$ hyperedges and $t \ge 2$. Then there exists integers $\delta_1,\dots,\delta_t\ge 0$ summing to $n-k$ such that taking $\delta_j$ copies of $S_j\defeq \{i \in [n]: j\notin e_i\}\subseteq [n]$ gives an $(n,n-k)$-GZP.
    \label{cor:hypergraph}
  \end{corollary}
  \begin{proof}
    Take the orientation of $\mathcal{H}$ and root vertex $v\in[t]$ given by Theorem~\ref{thm:hypergraph}.
    We now take our $\delta_j$'s as follows: for each non-root $j\in[t]$, let $\delta_j\defeq \deg_{in} (j)$ to be the in-degree of vertex $j$. 
    For the root $v$, let $\delta_v\defeq \deg_{in}(v)-k$.
    Note that any other vertex $u$ has $k$ edge-disjoint paths to $v$, so $v$ has in-degree at least $k$ and $\delta_v \ge 0$.
    Since there are $n$ hyperedges, the sum of all $\delta_j$'s is thus $n - k$.
    We now check the generic zero pattern condition \eqref{eq:gzp}.
    Consider any nonempty multiset $K\subseteq [t]$ such that each vertex $j \in [t]$ appears at most $\delta_j$ times. We claim:
    \begin{align}
      \abs{\bigcap_{\ell\in K} S_\ell} = (\text{\# edges induced by $[t]\setminus K$}) \le \sum_{j\in [t]\setminus K}^{} \delta_j = n-k - \sum_{j\in K\cap [t]}^{} \delta_j \le n-k-|K|.
      \label{eq:hypergraph-goal}
    \end{align}
    The first equality holds by definition of $S_j$.
    The second equality holds because $\sum_{j\in[t]}\delta_j=n-k$.
    The second inequality holds because $|K| \le \sum_{j \in K} \delta_j$ by definition of $K$.
    It remains to show the first inequality.
    We have two cases:
    
    \noindent\textbf{Case 1: the root $v$ is in $K$.}
    The number of hyperedges induced by the vertices $[t]\setminus K$ is at most the sum of the indegrees of $[t]\setminus K$, which is exactly $\sum_{j\in [t]\setminus K}^{} \delta_j$ by definition of $\delta_j$.
    
    \noindent\textbf{Case 2: the root $v$ is in $[t]\setminus K$.} 
    Fix an arbitrary vertex $u$ in $K$.
    By our orientation of $\mathcal{H}$, vertex $u$ has $k$ edge-disjoint paths to $v$. Each of these paths has an edge that ``enters'' $[t] \setminus K$, i.e., the head is in $[t] \setminus K$ but the edge is not induced by $[t]\setminus K$.
    Thus, the number of edges induced by $[t]\setminus K$ is at most $(\sum_{j\in[t]\setminus K}^{} \deg_{in}(j)) - k$, which is exactly $\sum_{j\in[t]\setminus K}\delta_j$ by definition of $\delta_j$.
    Hence, we have the first inequality.
    This covers all cases, proving \eqref{eq:hypergraph-goal}, completing the proof.
  \end{proof}

\subsection{Proof of Theorem~\ref{lem:int-mat-1}}
\label{subapp:int-mat-1-proof}

In this section, we reprove Theorem~\ref{lem:int-mat-1}, which we need in this work.

\begin{proof}[Proof of Theorem~\ref{lem:int-mat-1}]
  It suffices to prove that $\mathsf{RIM}_\mathcal{H}$ has full column rank for some evaluation of $X_1=\alpha_1,\dots,X_n=\alpha_n$ for $\alpha_1,\dots,\alpha_n\in\mathbb{F}_q$.
  Furthermore, by Remark~\ref{rem:int-mat-1}, it also suffices to prove Theorem~\ref{lem:int-mat-1} for when $q \ge 2n-k-1$.
  Indeed, that would then show there that is a square $(t-1)k\times (t-1)k$ submatrix of $\mathsf{RIM}_{\mathcal{H}}(X_{[n]}=\alpha_{[n]})$ of full column rank, which means that submatrix has nonzero determinant (in $\mathbb{F}_q$), which means the corresponding square submatrix of $\mathsf{RIM}_{\mathcal{H}}$ also has a nonzero determinant (in $\mathbb{F}_q(X_1,\dots,X_n)$), so $\mathsf{RIM}_{\mathcal{H}}$ has full column rank.
  
Let $e_1,\dots,e_n$ be the edges of our $k$-weakly-partition-connected hypergraph $\mathcal{H}$.
By Corollary~\ref{cor:hypergraph}, there a generic zero pattern $S_1,\dots,S_{n-k}$ where, for all $\ell=1,\dots,n-k$, the set $S_\ell$ is $\{i:j\notin e_i\}$ for some $j\in[t]$.
By Theorem~\ref{thm:gm-mds}, there exists pairwise distinct $\alpha_1,\dots,\alpha_n \in \mathbb{F}_q$ and a nonsingular matrix $M\in \mathbb{F}_q^{(n-k)\times (n-k)}$ such that, for $H\in\mathbb{F}_q^{(n-k)\times n}$ the parity check matrix of $\RS_{n,k}(\alpha_1,\dots,\alpha_n)$, the matrix $M\cdot H\in\mathbb{F}_q^{(n-k)\times n}$ achieves the zero pattern $S_1,\dots,S_{n-k}$,  meaning that $(MH)_{\ell,i}=0$ whenever $i\in S_\ell$.

  Suppose for the sake of contradiction there is a nonzero vector $v\in \mathbb{F}_q^{(t-1)k}$ such that $\mathsf{RIM}_{\mathcal{H}}(X_{[n]}=\alpha_{[n]}) \cdot v = 0$. Let $f\ind{1},\dots,f\ind{t}\in\mathbb{F}_q^{k}$ be such that $v=[f\ind{1},f\ind{2},\dots,f\ind{t-1}]$ and $f\ind{t}=0$.
  Define $c\ind{1},\dots,c\ind{t}\in\mathbb{F}_q^n$ be such that $c\ind{i} = G\cdot f\ind{i}$ where
  \begin{align}
    G\defeq
    \begin{bmatrix}
      1 & \alpha_1 & \cdots & \alpha_1^{k-1} \\
      1 & \alpha_2 & \cdots & \alpha_2^{k-1} \\
      \vdots &\vdots&&\vdots \\
      1 & \alpha_n & \cdots & \alpha_n^{k-1} \\
    \end{bmatrix}
  \end{align}
  We next show that, for any $i=1,\dots,n$, $c\ind{j}_i = c\ind{j'}_i$ for all $j,j'\in e_i$.
Let $e_i = j_1,\dots,j_{|e_i|}$.
Since $\mathsf{RIM}_{\mathcal{H}}(X_{[n]}=\alpha_{[n]}) \cdot v = 0$, we have, by definition of $\mathsf{RIM}_\mathcal{H}$, for $u=2,\dots,|e_i|$,
\begin{align}
  c\ind{j_1}_i-c\ind{j_u}_i = [1,\alpha_i,\dots,\alpha_i^{k-1}] \cdot (f\ind{j_1} - f\ind{j_u})^T = 0.
\end{align}
(note this is true even if $j_u = t$, since $f\ind{t}=0$).

Define a vector $y\in\mathbb{F}_q^n$ such that, for $i=1,\dots,n$, we have $y_i = c\ind{j}_i$, where $j$ is an arbitrary element of hyperedge $e_i$ (by the previous paragraph, the choice of $j$ does not matter).
  For each $j=1,\dots,t$, we must have $(MH\cdot (y-c\ind{j}))_\ell = 0$ for all $\ell \in [n-k]$ such that $S_\ell$ is a copy of $\{i \in [n] : j\notin e_i\}$; the $\ell$'th row of $MH$ is supported only on $\{i \in [n] : j\in e_i\}$, and $y-c\ind{j}$ is zero on $\{i \in [n] : j\in e_i\}$ by definition of $y$.
  Since $MHc\ind{j} = M\cdot(Hc\ind{j}) = 0$ for all $j=1,\dots,t$, we have, for all $j$ and all $\ell$ such that $S_\ell$ is a copy of $\{i \in [n] : j\notin e_i\}$,
  \begin{align}
    (MHy)_\ell = (MH\cdot (y-c\ind{j}))_\ell = 0.
  \end{align}
  By construction, all $S_\ell$ are a copy of some set $\{i:j\notin e_i\}$, so we conclude $MHy=0$.
  Since $M$ is invertible, we must have $Hy = 0$.

  This means $y = G\cdot f$ for some $f\in \mathbb{F}_q^k$, so $y$ is the evaluation of a degree-less-than-$k$ polynomial.
  Since $\mathcal{H}$ is $k$-weakly-partition-connected, by considering the partition $\{j\}\sqcup ([t]\setminus\{j\})$, there are at least $k$ hyperedges $e_i$ containing vertex $j$ in $\mathcal{H}$, so $y_i = c\ind{j}_i$ in at least $k$ indices $i$.
  Since $y$ and $c\ind{j}$ are the evaluation of degree-less-than-$k$ polynomials, we must have $y=c\ind{j}$. This holds for all $j$, so we have $y=c\ind{1}=\cdots=c\ind{t} = 0$ (recall $f\ind{t}=0$), which contradicts our initial assumption that $v\neq 0$.
\end{proof}

\section*{Acknowledgments} 

We thank Mary Wootters and Francisco Pernice for helpful discussions about \cite{BGM23} and the hypergraph perspective of the list-decoding problem.
We thank Karthik Chandrasekaran for helpful discussions about hypergraph connectivity notions and for the reference of Theorem~\ref{thm:hypergraph} in \cite{Frank11}.
We thank Nikhil Shagrithaya and Jonathan Mosheiff for pointing out a mistake in the proof of Lemma~\ref{lem:rlc:int-mat-1} in an earlier version of the paper.
We thank an anonymous reviewer for pointing out a mistake in the proof of Corollary~\ref{cor:hypergraph} in an earlier version of the paper.

\bibliographystyle{amsplain}
\bibliography{bib}


\begin{aicauthors}
\begin{authorinfo}[omara]
  Omar Alrabiah\\
  University of California, Berkeley\\
  Berkeley, California, United States of America\\
  oalrabiah\imageat{}berkeley\imagedot{}edu\\
  \url{https://scholar.google.com/citations?user=AZMV2qQAAAAJ}
\end{authorinfo}
\begin{authorinfo}[zeyug]
  Zeyu Guo\\
  The Ohio State University\\
  Columbus, Ohio, United States of America\\
  zguotcs\imageat{}gmail\imagedot{}com\\
  \url{https://zeyuguo.bitbucket.io}
\end{authorinfo}
\begin{authorinfo}[venkatg]
  Venkatesan Guruswami\\
  University of California, Berkeley and Simons Institute for the Theory of Computing\\
  Berkeley, California, United States of America\\
  venkatg\imageat{}berkeley\imagedot{}edu\\
  \url{https://people.eecs.berkeley.edu/~venkatg}
\end{authorinfo}
\begin{authorinfo}[rayl]
  Ray Li\\
  Santa Clara University\\
  Santa Clara, California, United States of America\\
   rli6\imageat{}scu\imagedot{}edu\\
   \url{https://cs.stanford.edu/~rayyli}
\end{authorinfo}
\begin{authorinfo}[zihanz]
  Zihan Zhang\\
  The Ohio State University\\
  Columbus, Ohio, United States of America\\
  zhang.13691\imageat{}osu\imagedot{}edu \\
  \url{https://zihanzhang.owlstown.net}
\end{authorinfo}
\end{aicauthors}

\end{document}